%% file: Position_latex-one-column.tex
\documentclass[11pt,draftcls,onecolumn]{IEEEtran}
\input{preamble.tex}

\input{mydefinitions.tex}
\usepackage{setspace}
\usepackage{autobreak}
\begin{document}

%add the numbering of the article
%\pagenumbering{arabic}

\title{Using Reconfigurable Intelligent Surfaces for UE Positioning in mmWave MIMO Systems}

\author{Wei Zhang,~\IEEEmembership{Member,~IEEE}
and Wee Peng Tay,~\IEEEmembership{Senior Member,~IEEE}

\thanks{
The authors are with the School of Electrical and Electronic Engineering, Nanyang Technological University, Singapore (e-mail: weizhang@ntu.edu.sg, wptay@ntu.edu.sg).
This research is supported by A*STAR under its RIE2020 Advanced Manufacturing and Engineering (AME) Industry Alignment Fund – Pre Positioning (IAF-PP) (Grant No. A19D6a0053).
Part of this work has been submitted to the IEEE International Conference on Communications, 2022.
}
}

\maketitle
\begin{abstract}
A reconfigurable intelligent surface (RIS) consists of massive meta elements, which results in a reflection path between a base station (BS) and user equipment (UE). In wireless localization, this reflection path aids in positioning accuracy, especially when the line-of-sight (LOS) path is subject to severe blockage and fading. We develop a RIS-aided positioning framework to locate a UE in environments where the LOS path may or may not be available. We first estimate the RIS-aided channel parameters from the received signals at the UE. To reduce algorithmic complexity, we propose a linear combination of the estimated UE positions from the direct and reflection paths, which is shown to be approximately the maximum likelihood estimator under the large-sample regime when the estimates from different paths are independent. We optimize the RIS phase shifts to improve the positioning accuracy, and extend the proposed approach to the case with multiple BSs and UEs. We derive the Cramér–Rao bound (CRB) and demonstrate numerically that our proposed method approaches the CRB.
\end{abstract}

\section{Introduction}

The reconfigurable intelligent surface (RIS) has been proposed as an aid to wireless communication systems.
A RIS consists of many low-cost meta elements \cite{DongSur2020,Basar2019WirelessAccess}, through which the performance of existing wireless communication systems can be improved without significant additional hardware cost.
Different from a relay, a RIS passively reflects the received signal, changing only its phase shift before transmission to a user equipment (UE) \cite{RenzoSur2020,EmilBeat2020,BoulSur2020}.
With the help of a RIS, a reflected transmission path can be established if the direct transmission path is blocked, which makes the RIS potentially useful for urban or indoor communications \cite{Liang2019Large}.
%RIS also has the potential for aiding positioning and localization \cite{He2020,Sha2018}.
%By carefully deploying the RIS, RIS also helps to improve the performance of cell-edge users \cite{PanMul2020,XieMax2020}.
%There are also some works focus on the channel estimation when the RIS is involved \cite{Nad2020Intelli, Zheng2019Channel}.
Compared to traditional wireless communication with transmit beamforming, the phase shifts of a RIS can be configured to achieve passive beamforming \cite{ChenSum2019,ZhaoRIS2021,NadeemCE2020} for RIS-aided systems.
%When the RIS is introduced, some existing works focus on the channel estimation when the RIS is involved \cite{Nad2020Intelli, Zheng2019Channel}.
With properly designed passive beamforming, many works in the literature have shown that the RIS can improve various system performance metrics, such as spectral efficiency \cite{ChenSum2019,NadeemCE2020},  received signal to noise ratio (SNR) \cite{Basar2019Tran,Atapattu2020RIS} and bit error rate \cite{Ferreira2020Bit}. These works suggest that RIS can play an important role in future wireless communication systems.

In this paper, we investigate the use of RIS in user equipment (UE) positioning. Localization or positioning is an important task in wireless communications \cite{Yassin2017Localization,WEN2019Survey}. In 5G systems, positioning of the UE has diverse applications, including industrial use cases, smart mobility, and location-based services. The use of 5G millimeter wave (mmWave) has the potential to provide better positioning accuracy compared to the Global Positioning System (GPS) \cite{Wymeersch2017mmWave}. As such, the 3GPP Release 16 \cite{3gpp.21.916} has incorporated standards for location management in the 5G NR framework.

%and the accuracy is determined by the signal strength model. The difficulty of using RSS-based methods is from the fact that the signal strength is corrupted by small-scale fading, which can be reduced by averaging among the time domain.
In \cite{Vari2014mmV}, a positioning method based on received signal strength (RSS) of mmWaves was presented. By using sufficient measurements and multiple access points (APs), the position of a UE can be estimated by the trilateration technique, which utilizes the ranges estimated through measuring the received signal power. In \cite{Lin2018indoor}, by using one AP, a hybrid RSS and angle of arrival (AoA)
positioning scheme is developed to provide estimates of both the distance and the orientation of the target. In \cite{Parasas2018ML}, the authors train a  noise-free RSS model and then use it to position UEs with noisy RSS.
However, for the RSS-based methods, the positioning accuracy is determined by the signal strength model, and the signal strength is often corrupted by small-scale fading, which is challenging to estimate, especially when the measurement time of the signal is insufficient.

Apart from the RSS-based methods, some works \cite{Fas2019mmW,Zhou2017low,Wen2021Pos} obtain the position of a UE in mmWave systems by estimating the time of arrival (ToA) and AoA or angle of departure (AoD). Therefore, the positioning task can be treated as a channel estimation problem.
%Due to the short wave length, massive antennas can be equipped both at the transmitter and receiver to provide beamforming gain, which also impose challenge for channel estimation.
Because of the limited scattering of paths in mmWave communication, compressed sensing (CS) methods can be employed to reduce the measurement time and computational complexity \cite{Alkhateeb2014Channel,Zhang2018RIP}.
%The sparse property means that number of paths are limited, thus the channel is sparsity in angular domain. Therefore, the angle estimation of millimeter system can be employed for positioning.
Specifically,
in \cite{Fas2019mmW}, the \gls{MLE} for time of delay and AoD is discussed under the line-of-sight (LOS) scenario with the base station (BS) being equipped with massive antennas and the UE having one antenna.
In \cite{Shah2018Position}, the authors estimate the position and rotation angle of the UE using a single BS, and the proposed method can be applied in a non-LOS (NLOS) environment.
In \cite{Zhou2017low}, the received signal measurements are structured as a tensor, based on which the channel parameters such as ToAs, AoAs and AoDs are extracted.
 In \cite{Wen2021Pos}, a tensor-based channel estimation method for positioning and mapping was proposed for diffuse multipaths.

Since a RIS creates a reflection path between a BS and UE, the UE can utilize the measurements from this reflection path as additional information for positioning. Some works have shown that the positioning accuracy improves with the size of the RIS. The Cram\'er-Rao lower bound (CRB) of the positioning accuracy is analyzed in \cite{Hu2018Beyond,Aleg2019CRLB,Ji2020large,Ma2021indoor,Elzanaty2021Re}.
%For the positioning task, there are some existing works talking about the techniques of positioning when RIS is introduced.
%However, when then LOS path is heavily blocked, the traditional GPSs-based or 5G-based positioning will not accurate positioning result.
%Since the RIS can provide reflection path when the direct path is blocked by obstacles or in indoor scenarios, in which the accuracy will become bad. The invention of RIS can help to solve the issue of blockage.
However, few existing literature have developed practical positioning algorithms for a RIS-aided system. Indoor positioning using the RSS is investigated by \cite{zhang2021toward,Zhang2021Meta}, which estimates the position of a UE using the probability distribution of the RSS.
In \cite{Lin2021Channel}, the authors consider channel estimation and geometric mapping for positioning under the twin-RIS scenario.

In this paper, we develop a novel positioning and inference framework for RIS-aided systems using channel estimation techniques.
Our approach is not limited to using the RSS measurements. Different from the existing works \cite{zhang2021toward,Zhang2021Meta}, we formulate our problem under the general case where there may be more than one RIS. In contrast with existing RIS-aided channel estimation methods \cite{Jensen2020,De2021Channel,Zhang2021Cost} that estimate the cascaded channel by assuming that the direct channel is estimated in advance, we estimate the channel parameters such as ToAs, AoAs and AoDs of the direct and reflection paths jointly.
In addition, different from the geometric mapping in \cite{Lin2021Channel}, our proposed inference model considers the estimation accuracy of the channel parameters, which yields a UE position estimation error close to the theoretical CRB.

The main contributions of this paper are summarized as follows:
\begin{itemize}
  \item We consider the down-link MIMO-OFDM setup in this work.
  Direct estimation of the UE position from the received signals is computationally expensive as it involves a nonlinear and non-convex optimization. Therefore,  we propose a two-step positioning framework. In the first step, we estimate the channel parameters of the direct and reflection paths. In the second step, we obtain an estimate of the UE position from the channel parameters of each path. We derive the CRB of the UE position estimate under our positioning framework.

  \item  To infer the UE position from the different estimates corresponding to the direct and reflection paths, we perform a linear combination of these estimates. The linear combination weights depend only on the covariance of the UE position estimates. We show when the estimates from different paths are independent, the proposed linear combination is approximately the \gls{MLE} of the UE position in the large-sample regime.

  \item To optimize the positioning framework, we propose an approach for designing the RIS phase shifts. Specifically, the phase shift design problem is to maximize the expectation of the reflection path gain, which can be then solved using singular value decomposition.

\end{itemize}

One challenge is to distinguish the direct and reflection paths. In this work, different from the existing works where the path with the smallest delay is assigned as the direct path, we distinguish the direct and the reflection paths by ranking a path quantity related to its power level. This method is more robust if the SNR is low. Our proposed RIS-aided positioning framework is also readily extended to the multi-UE and multi-BS scenarios.

The rest of this paper is organized as follows. In \cref{sect:signal model}, the signal and channel model, and our system assumptions are introduced.
In \cref{sec:CRB}, we derive the CRB of the UE positioning error under the signal and channel model.
The proposed RIS-aided channel parameter estimation approach is discussed in \cref{sec:estimation}. In \cref{sec:fusion}, we propose the fusion method to infer the UE position from the estimated channel parameters. In \cref{sec:discussions}, we propose the method to optimize the RIS phase shifts and discuss the extension of our positioning framework to the multi-UE and multi-BS scenarios.
We present numerical results in \cref{sec:numerical}. Finally, we conclude in \cref{sec:conclusions}.
%</tag:QQ1>

\emph{Notations:} A bold lower case letter $\mathbf{a}$ is a vector and a bold capital letter $\mathbf{A}$ represents a matrix. $\bA\T$, ${{\mathbf{A}}^{H}}$, ${{\mathbf{A}}^{-1}}$, $\mathrm{tr}(\mathbf{A})$, $\left| \mathbf{A} \right|$,  ${{\left\| \mathbf{A} \right\|}_{F}}$ and ${{\left\| \mathbf{a} \right\|}_{2}}$ are, respectively, the transpose, Hermitian, inverse, trace, determinant, Frobenius norm of $\bA$, and the $2$-norm of $\mathbf{a}$.
${{[\mathbf{A}]}_{:.i}}$, ${{[\mathbf{A}]}_{i,:}}$, ${{[\mathbf{A}]}_{i,j}}$, and $[\ba]_i$ are, respectively, the $i$th column, $i$th row, $i$th row and $j$th column entry of $\mathbf{A}$, and the $i$th entry of vector $\mathbf{a}$. The operation $\mathrm{\mathop{vec}}(\mathbf{A})$ stacks the columns of $\mathbf{A}$ to form a column vector.
$\Col(\bA)$ is the column space of matrix $\bA$.
We use $\diag(\ba)$ to represent a diagonal matrix with the vector $\ba$ on the main diagonal.
The circular symmetric complex Gaussian distribution with mean $\mu$ and variance $\sigma^2$ is given by $\dist{\calC\calN}[\mu, \sigma^2]$. We use $\otimes$ to denote the Kronecker product.
 %${{\mathbf{I}}_{M}}\in {{\mathbb{R}}^{M\times M}}$ is the $M$-dimensional identity matrix.

\section{System Model} \label{sect:signal model}

In this section, we present our system model and assumptions. We first discuss the channel model, which includes the BS-RIS links, RIS-UE links and BS-UE link. We then present the received signal model at the UE.

\begin{figure}[!htpb]
\centering
\includegraphics[width=.6\textwidth]{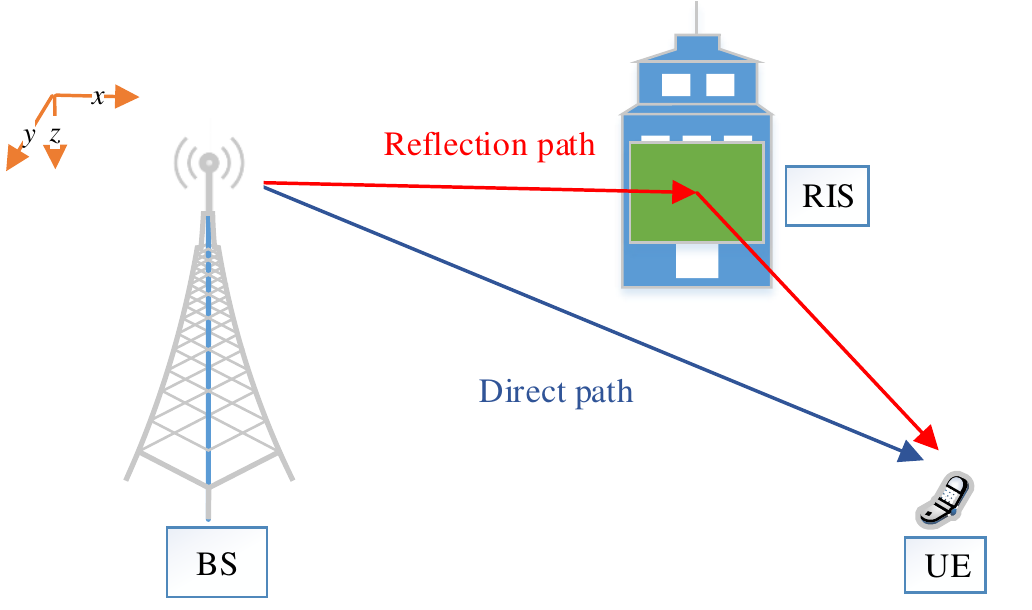}
\caption{Positioning of a UE with the aid of a RIS. } \label{fig:system}
\end{figure}

\subsection{Channel Model}
We assume that the BS has a uniform rectangular array (URA) with $N$ antennas. There are $Q$ RISs and each is equipped with a URA of $M$ elements. The UE has $D$ antennas. In this work, we assume that the position of every RIS is known by the BS and UE. Without loss of generality, we adopt a coordinate system with the BS at its origin and the URA of the BS in $y-z$ plane (see \cref{fig:system} for an illustration). Each RIS' URA is assumed to be contained in a $x-z$ plane perpendicular to the $y-z$ plane of the BS URA.

We also assume that the UE's antennas are contained in a horizontal plane parallel to the BS URA, but with a possibly different orientation. Let $\widebar{\bM}_R\in \R^{3 \times 3}$ be the rotation matrix associated with the UE, given by
\begin{align}
\widebar{\bM}_R &= -\bR_3(\alpha_3)\bR_2(\alpha_2)\bR_1(\alpha_1)\nonumber \\
&=
-\begin{bmatrix}
\cos \alpha_3 & \sin \alpha_3 &0\\
 -\sin \alpha_3 & \cos \alpha_3 &0\\
 0&0&1\\
\end{bmatrix}
\times
 \begin{bmatrix}
1&0&0\\
0& \cos \alpha_2 & \sin \alpha_2 \\
 0& - \sin \alpha_2 &  \cos \alpha_2
\end{bmatrix}
\begin{bmatrix}
\cos \alpha_1 & \sin \alpha_1 &0\\
 -\sin \alpha_1 & \cos \alpha_1 &0\\
 0&0&1\\
\end{bmatrix} ,\label{eq:expression barMR}
\end{align}
where $\alpha_1,\alpha_2,\alpha_3$ are the Euler angles \gls{wrt} the UE. For convenience, we define $\bM_R =[\widebar{\bM}_R]_{2:,3,:}$.
In this work, we assume that $\widebar{\bM}_R$ is known a priori by the UE.

We suppose that the communication system uses OFDM with $K$ subcarriers. For the $k$th subcarrier, the channel from the BS to the $q$th RIS is denoted as $\bG_{k,q}\in \C^{M \times N}$, the channel from the UE to the $q$th RIS is $\bH_{r,k,q} \in \C^{D\times M}$, and the channel from the BS to the UE is $\bH_{d,k} \in \C^{D\times M}$.

\subsubsection{BS-RIS links}

In this work, we model the BS-RIS channel as a mmWave channel. We assume that each RIS is placed at a sufficient height (e.g., on a tall building) so that there is a LOS path between the BS and the RIS.

%In the simulation part, we also illustrate the proposed method in the case when the NLOS paths are introduced. In this
From the OFDM assumption, the $k$th subcarrier of the $q$th BS-RIS channel is \cite{He2020Cascaded,chen2019channel}
\begin{align}
\! \! \!  \bG_{k,q} \!\!  = \!\!  h_{R_1,q}e^{\!-\iu 2\pi \frac{kW}{K} \tau_{r_1,q}} \ba_R(f_{R_1,q}, v_{R_1,q}\!  )\ba_B^H(g_{B_r,q},v_{B_r,q}\!), \! \! \label{eq:Channel_model_BS_RIS}
\end{align}
where $\iu = \sqrt{-1}$, $ h_{R_1,q} = \alpha_{R_1,q}\sqrt{\beta_{R_1,q}MN}$ with $\beta_{R_1,q}$ being the large scale path gain and $\alpha_{R_1,q}$ being a complex-valued channel coefficient. $W$ is the transmission bandwidth, and $\tau_{r_1,q}$ is the propagation delay of the signal from the BS to the $q$th RIS.
In particular, $\ba_R(f_{R_1,q}, v_{R_1,q}) \in \C^{M\times 1}$ and  $\ba_B(g_{B_r,q},v_{B_r,q})\in \C^{N\times 1}$ are, respectively, the URA response vectors of the RIS and BS, where
\begin{align}
f_{R_1,q} &= \sin\theta_{R_1,q} \cos \phi_{R_1,q}, ~v_{R_1,q} = \cos\theta_{R_1,q} , \label{eq:fR1q}\\
g_{B_r,q} &= \sin\theta_{B_r,q} \sin \phi_{B_r,q},~v_{B_r,q} = \cos \theta_{B_r,q}, \label{eq:fBrq}
\end{align}
with the $\theta_{R_1,q}$ (or $\theta_{B_r,q}$)  and  $\phi_{R_1,q}$ (or $\phi_{B_r,q}$) being the elevation and azimuth AoAs (or AoDs) associated the BS-RIS link, respectively. To be more precise, the URA response vectors $\ba_R(f, g)$ and $\ba_B(f, g)$ in \cref{eq:Channel_model_BS_RIS} are given by
\begin{align}
\ba_R(f,v) &=  \ba_{\tilde{R}}(f) \otimes \ba_{\tilde{R}}(v), \label{eq:aR}\\
\ba_B(g, v) &=  \ba_{\tilde{B}}(g) \otimes \ba_{\tilde{B}}(v), \label{eq:aB}
\end{align}
where
$
\ba_{\tilde{R}}(f) = \frac{1}{M^{1/4}}\left[ 1, \exp({\iu \pi f}), \ldots,\exp({\iu \pi f(M^{1/2}-1)}) \right]
$
 and
$\ba_{\tilde{B}}(g) = \frac{1}{N^{1/4}}\left[ 1, \exp({\iu \pi g}),\ldots,\exp({\iu \pi g(N^{1/2}-1)}) \right]$.

% \iffalse
% There are some constraints for the angle domain. For BS, $\phi_{B_r,q} \in (-\pi/2,\pi/2)$, which means that $\cos \phi_{B_r,q}\ge 0$. For RIS, $\theta_{R,1,q}  \in (\pi/2,\pi)$, which means that $\cos \theta_{R,1,q}\le 0$.
% Assume the position of the user is $[x_u,y_u,z_u]$, and the position of RIS is $[x_r, y_r, z_r]$. The requirement for the angles requires that
% \begin{align}
% x_r>0, z_r>0.
% \end{align}

% \begin{Remark}
% In this work, we assume the position of RIS is known by the BS and user. Thus, the channel associated with BS-RIS can be known as a priori.
% \end{Remark}
% \iffalse
% \begin{align}
% \ba_R(\theta_{R,1,q}, \phi_{R,1,q}) =\ba_{\tilde{R}}(\sin\theta_{R,1,q}\sin \phi_{R,1,q}) \otimes \ba_{\tilde{R}}(\cos\theta_{R,1,q}) \nonumber \\
% \ba_B(\theta_{B_r,q}, \phi_{B_r,q}) =\ba_{\tilde{B}}(\sin \theta_{B_r,q} \sin \phi_{B_r,q} ) \otimes \ba_{\tilde{B}}(\cos\theta_{B_r,q})\nonumber.
% \end{align}
% \fi
% \fi

\subsubsection{RIS-UE link}

For the channel between the $q$th RIS and the UE, we again assume that a LOS path exists between the RIS and the UE. The $k$th subcarrier channel of the RIS-UE link is given by
\begin{align}
\!\!\!  \bH_{r,k,q} \!\!= \!\!h_{R_2,q} e^{-\iu 2\pi \frac{kW}{K}\!\tau_{r_2,q}}\ba_U\!(g_{U_r,q}, \! v_{U_r,q})\ba_R^H\!(f_{R_2,q}, \!v_{R_2,q})\!,\!  \!\! \label{eq:Channel_model_RIS_UE}
\end{align}
where $h_{R_2,q} = \alpha_{R_2,q}\sqrt{\beta_{R_2,q}MD} $ with $\beta_{R_2,q}$ being the large scale path gain and $\alpha_{R_2,q}$ being complex-valued channel coefficient, and $\tau_{r_2,q}$ is the delay.
The URA response vector $\ba_R(f_{R_2,q},v_{R_2,q})$ is given in \cref{eq:aR} and $\ba_U(g_{U_r,q}, v_{U_r,q})\in \C^{D\times 1}$ is the URA response vector of the UE, where
\begin{align}
& f_{R_2,q} = \sin\theta_{R_2,q} \sin \phi_{R_2,q}, ~v_{R_2,q} = \cos \theta_{R_2,q} ,\label{AoA:U_r_q 1} \\
&
\begin{bmatrix}
 g_{U_r,q}\\
v_{U_r,q}
\end{bmatrix}=
\bM_R
\begin{bmatrix}
\sin\theta_{R_2,q} \cos \phi_{R_2,q}\\
\sin\theta_{R_2,q} \sin \phi_{R_2,q}\\
\cos \theta_{R_2,q}
\end{bmatrix}, \label{AoA:U_r_q}
\end{align}
with  $\theta_{R_2,q}$  and  $\phi_{R_2,q}$ being the elevation and azimuth AoDs associated with the RIS-UE link. Abusing terminology, we refer to $(f_{R_2,q},v_{R_2,q})$ as the AoD of the $q$th RIS, and $(g_{U_r,q},v_{U_r,q})$ as the AoA of the UE on the reflection path.

The URA response vector of the UE is
\begin{align}
\ba_U(g, v)=\ba_{\tilde{U}}(g) \otimes \ba_{\tilde{U}}(v), \label{eq:aU}
\end{align}
where $\ba_{\tilde{U}}(g) = \frac{1}{D^{1/4}}\left[ 1, \exp({\iu \pi g}), \ldots,\exp({\iu\pi g(D^{1/2}-1)}) \right] \in \C^{D^{1/2} \times 1}$. For example, the response due to the $k$th subcarrier on the RIS-UE link is given by $\ba_U(g_{U_r,q},v_{U_r,q})$.

% \iffalse
% There are some constraints for the angle domain. For RIS, $\theta_{R_2,q}\in (\pi/2,\pi)$, which means that $\cos \theta_{R_2,q}\le 0$. For user, $\phi_{U_r,q} \in (\pi/2,3\pi/2)$, which means that $\cos \phi_{U_r,q}\le 0$.
% Assume the position of the user is $[x_u,y_u,z_u]$, and the position of RIS is $[x_r, y_r, z_r]$. The requirement for the angles requires that
% \begin{align}
% z_r-z_u>0, x_r-x_u<0. \nonumber
% \end{align}
% \fi

\subsubsection{BS-UE link}
We model the BS-UE link channel using the Rician fading model, given by
\begin{align}
\bH_{d,k} =\widebar{\bH}_{d,k} + \underbrace{\sqrt{ {1}/{(1+K_d)}}\bZ_{d,k}}_{=\tilde{\bZ}_{d,k}}, \label{eq:Channel_model_BS_UE}
\end{align}
where $K_d$ is the Rician factor, $\widebar{\bH}_{d,k}$ is the deterministic component or the LOS path, and $\bZ_{d,k}$ denotes the small-scale fading
whose entries are \gls{iid} according to $\cC\cN(0,\beta_d)$ with $\beta_d$ being the large scale path gain. The expression of $\widebar{\bH}_{d,k}$ is given by
\begin{align}
\widebar{\bH}_{d,k} =h_{d} e^{-\iu2\pi \frac{kW}{K}\tau_d}\ba_U(g_{U_d}, v_{U_d})\ba_B^H(g_{B_d},v_{B_d}),
\end{align}
where we let $h_{d} =\sqrt{ \frac{K_d}{1+K_d}\beta_{d}ND}  \alpha_d$ with $\alpha_{d}$ being complex-valued channel coefficient, and the URA response vectors of the UE and the BS, $\ba_U(g_{U_d}, v_{U_d})$ and $\ba_B(g_{B_d}, v_{B_d})$ are defined in \cref{eq:aU} and \cref{eq:aB}, respectively. We have
\begin{align}
& g_{B_d} =\sin \theta_{B_d}\sin \phi_{B_d}  , ~v_{B_d}=\cos\theta_{B_d}, \label{AoD:B_d}\\
&
\begin{bmatrix}
g_{U_d}\\
v_{U_d}
\end{bmatrix}
=
\bM_R
\begin{bmatrix}
\sin\theta_{B_d} \cos \phi_{B_d}\\
\sin\theta_{B_d} \sin \phi_{B_d}\\
\cos \theta_{B_d}
\end{bmatrix}, \label{AoD:U_d}
\end{align}
 where $\theta_{B_d}$  and  $\phi_{B_d}$ are the elevation and azimuth AoDs associated the BS-UE link. Abusing terminology, we refer to $(g_{B_d},v_{B_d})$ as the AoD of the BS, and $(g_{U_d},v_{U_d})$ as the AoA of the UE on the LOS path.

In summary, using the channel models of the BS-RIS link in \cref{eq:Channel_model_BS_RIS}, the RIS-UE link in \cref{eq:Channel_model_RIS_UE}, and the BS-UE link in \cref{eq:Channel_model_BS_UE}, the effective channel between the BS and UE on the $k$th subcarrier can be written as
\begin{align}
\bH_k &= \widebar{\bH}_{d,k}+\sum_{q=1}^{Q} \bH_{r,k,q} \bTheta_q \bG_{k,q} + \tilde{\bZ}_{d,k} \label{eq:final expression of Hk}\\
&=  h_{d} e^{-\iu 2\pi \frac{kW}{K}\tau_d} \ba_U(g_{U_d}, v_{U_d})\ba_B^H(g_{B_d},v_{B_d}) + \tilde{\bZ}_{d,k}\nonumber  \\
&~~~+\sum_{q=1}^{Q} h_{r,q} e^{-\iu 2\pi \frac{kW}{K}(\tau_{r_1,q} + \tau_{r_2,q})}\ba_U (g_{U_r,q}, v_{U_r,q })\ba_B^H(g_{B_r,q}, v_{B_r,q}),\nonumber
\end{align}
where $\bTheta_q=\diag(\btheta_q)$ with $\btheta_q = [e^{\iu\theta_q^{(1)}}, \ldots,e^{\iu\theta_q^{(M)}}]$ denoting the phase shift of the $q$th RIS, and $h_{r,q}=h_{R_1,q}h_{R_2,q}\ba_R^H(f_{R_2,q},v_{R_2,q})\bTheta_q \ba_R(f_{R_1,q},v_{R_1,q})$.
For convenience, we denote $\widebar{\bH}_k =\widebar{\bH}_{d,k} +  \sum_{q=1}^{Q} \bH_{r,k,q} \bTheta_q \bG_{k,q}$ in \cref{eq:final expression of Hk}.

Here, we define the channel parameters as
\begin{align}\label{eq:bmeta}
\boldsymbol{\eta}=[\boldsymbol{\eta}_d\T,\boldsymbol{\eta}_{r,1}\T, \ldots, \boldsymbol{\eta}_{r,Q}\T] \T\in \R^{(7+5Q)},
\end{align}
where
$\boldsymbol{\eta}_d =[\opRe\{h_d\}, \opIm\{ h_d \},\tau_d,g_{U_d},v_{U_d},g_{B_d},v_{B_d}]\T$,
$\boldsymbol{\eta}_{r,q} =[\opRe\{h_{r,q}\},\opIm\{h_{r,q}\}, \tau_{r_2,q},g_{U_r,q},v_{U_r,q}]\T.$
We denote the position of the UE as $\bp_U=[x_U,y_U,z_U]\T$, and the position of the $q$th RIS as $\bp_{R,q}=[x_{R,q},y_{R,q},z_{R,q}]\T$. To relate the channel parameters $\bmeta$ to the UE position, let
\begin{align}\label{eq:bxi}
\bxi=[\bp_U\T,\opRe\{h_d\}, \opIm\{h_d\},\opRe\{h_{r,1}\}, \opIm\{h_{r,1}\},
 \ldots,\opRe\{h_{r,Q}\}, \opIm\{h_{r,Q}\}]\T.
\end{align}
Then, we can define a function $F(\bxi)= \bmeta$ from the relations of \cref{eq:fR1q}, \cref{eq:fBrq}, \cref{AoA:U_r_q 1}, \cref{AoA:U_r_q}, \cref{AoD:B_d}, \cref{AoD:U_d}, and  the following equalities:
\begin{align}
\begin{aligned}\label{eq:F}
\tau_{d}  &={\| \bp_U \|_2}/{c},\
\tau_{r_2}  = {\| \bp_U-\bp_R \|_2}/{c}\\
\theta_{B,d} &= \arccos{\frac{z_U}{\| \bp_U \|_2}},
\phi_{B,d}  =  \arctan 2(y_U,x_U),\\
\theta_{R_2,q} &= \arccos{\frac{z_U-z_{R,q}}{\| \bp_U-\bp_{R,q} \|_2}},\\
\phi_{R_2,q}   &=  \arctan 2(y_U-y_{R,q},x_U-x_{R,q}).
\end{aligned}
\end{align}

\subsection{Received Signal at the UE}
Suppose that the UE receives signals over $T$ time slots. From the channel model \cref{eq:final expression of Hk}, the received signal at the UE at each time $t=1,\ldots,T$ on the $k$th subcarrier is given by
\begin{align}
  \br_k(t) &= \bH_k \bx (t)  +  \bn_k(t), \label{eq:t receive signal}
\end{align}
where $\bx_k(t) \in \C^{N \times 1}$ is the transmitted signal from the BS at time $t$, and $\bn_k(t) \in \C^{D \times 1}$ is a noise vector with entries \gls{iid} according to the complex Gaussian distribution $\dist{\cC\cN}[0, \sigma^2 ]$ and independent across time. Let
$\bR_k = [\br_k(1),\ldots, \br_k(T)] \in \C^{D\times T}$,
$\bX= [\bx(1),\ldots, \bx(T)] \in \C^{D\times T}$, and
$\bN_k = [\bn_k(1),\ldots, \bn_k(T)] \in \C^{D\times T}$.
We assume that the transmitted signals are orthogonal, i.e., $\bX\bX^H = {T}/{D}\bI$, where $\bI$ is the identity matrix.
Moreover, the transmit power is assumed to be unit, i.e., $\|\bx(t)\|_2^2 = 1$, for $t=1,\ldots,T$. The compact form of the received signal in \cref{eq:t receive signal} is given by
\begin{align}
\bR_k = \bH_k \bX + \bN_k. \label{eq:compact signal}
\end{align}
Right multiplying \cref{eq:compact signal} by ${D}/{T}\bX^H$, we have
\begin{align}
({D}/{T})\bR_k\bX^H  =\bH_k+ ({D}/{T})\bN_k\bX^H. \label{eq:mid receive signal}
\end{align}
The entries in $({D}/{T})\bN_k\bX ^H$ are \gls{iid} Gaussian $\cC \cN(0, \sigma^2  D/T)$ random variables. Here, we define $\tilde{\bR}_k = ({D}/{T})\bR_k\bX^H$ and recalling the definition of $\bH_k$ in \cref{eq:final expression of Hk}, we obtain
\begin{align}
\tilde{\bR}_k  = \widebar{\bH}_k + \tilde{\bZ}_{d,k} + ({D}/{T})\bN_k\bX^H
=\widebar{\bH}_k + \tilde{\bN}_k, \label{eq:R observations}
\end{align}
where we denote $\tilde{\bN}_k =\tilde{\bZ}_{d,k}+(D/T)\bN_k\bX^H$, and its entries follows $\cC\cN(0, \tilde{\sigma}^2 )$ with $ \tilde{\sigma}^2=\dfrac{D}{T}\sigma^2+\dfrac{\beta_d}{1+K_d}$.

Our objective is to infer the position of the UE by using the observations $\{ \tilde{\bR}_k \}_{k=1}^K$ in \cref{eq:R observations}. Because directly estimating the UE position from \cref{eq:R observations} is challenging, we first estimate the channel parameters, from which the UE position is then inferred.

\section{CRB for UE Position Estimation}\label{sec:CRB}
In this section, we derive the CRB for the UE position estimation based on the observations in \cref{eq:R observations}. We will compare the performance of the proposed method against this bound in the numerical results in \cref{sec:numerical}.

\subsection{FIM of the channel parameters \txp{$\bmeta$}{eta}}
Recall that our observations are $\tilde{\bR}_k  =\widebar{\bH}_k+ \tilde{\bN}_k$ in \cref{eq:R observations}.
We perform two steps to obtain the Fisher information matrix (FIM).
In the first step, we compute the FIM \gls{wrt} $\boldsymbol{\eta}$ of \cref{eq:bmeta}.
For any unbiased estimator $\boldsymbol{\hat{\eta}}$, we have
\begin{align}
\Eb[(\boldsymbol{\hat{\eta}}-\boldsymbol{\eta}) (\boldsymbol{\hat{\eta}}-\boldsymbol{\eta})\T] \succeq \left(\sum_{k=0}^{K-1}\bF_{\bmeta}^{(k)}\right)^{-1},
\end{align}
where $\bF_{\bmeta}^{(k)} \in  \R^{(7+5Q)\times (7+5Q)}$ is the FIM of $\bmeta$ based on the observations from the $k$th subcarrier.
Accordingly, the FIM of $\bmeta$ based on the observations from all the $K$ subcarriers is
\begin{align}
 \bF_{\bmeta}= \sum_{k=0}^{K-1} \bF^{(k)}_{\bmeta}. \label{eq:FIM channel parameters}
\end{align}
Because the noise in \cref{eq:R observations} is Gaussian, we have the following
\begin{align}
\ln f(\tilde{\bR}_k| \boldsymbol{\eta}) =  -\frac{1}{ \tilde{\sigma}^2}\tr((\tilde{\bR}_k-\widebar{\bH}_k)^H (\tilde{\bR}_k-\widebar{\bH}_k)) + C,
\end{align}
where $C$ is a normalization constant. The $(i,j)$-th element of $\bF_{\bmeta}^{(k)}$ is then given by

\begin{align*}
[\bI^{(k)}_{\bmeta}]_{i,j} & =-\Eb\left[\frac{\partial^2 \ln f(\tilde{\bR}_k| \boldsymbol{\eta})}{\partial {\eta}_i \partial {\eta}_j}\right].
\end{align*}

After simplifications, we have
\begin{align}
[\bI^{(k)}_{\bmeta}]_{i,j} = \frac{2}{ \tilde{\sigma}^2}\real \left\{ \tr\left( \frac{\partial \widebar{\bH}_k^H}{\partial \eta_i} \frac{\partial \widebar{\bH}_k}{\partial \eta_j} \right) \right\}.
\end{align}
\Cref{app:CRB} provides detailed derivation for terms in the FIM.

\subsection{FIM for the UE position parameters \txp{$\bxi$}{bxi}} \label{section FIM p}
% \iffalse
% We denote the position of user is given by $\bp_U = [x_U,y_U,z_U]\T$, and the position of $q$th RIS is $\bp_{R,q} = [x_{R,q},y_{R,q},z_{R,q}]\T$. The following show the relations between the user position and the channel parameters,
% \begin{align}
% \tau_d&=\| \bp_U \|_2, ~\tau_{r_2,q}=\| \bp_U  - \bp_{R,q}\|_2\nonumber\\
% \theta_{R_2,q}&= \arccos \left(   \frac{z_U-z_{R,q}}{\| \bp_U  - \bp_{R,q} \|_2}   \right)   ,
% \phi_{R_2,q}   =   \arctan 2   \left(   \frac{y_U-y_{R,q}}{x_U-x_{R,q}}   \right)\nonumber\\
% f_{R_2,q} &= \sin\theta_{R_2,q} \sin \phi_{R_2,q}, ~g_{R_2,q} = \sin\theta_{R_2,q} \cos \phi_{R_2,q} \nonumber\\
% \begin{bmatrix}
% f_{U_r,q}\\
% g_{U_r,q}
% \end{bmatrix}
% &=
% \bM_R
% \begin{bmatrix}
% \sin\theta_{R_2,q} \cos \phi_{R_2,q}\\
% \sin\theta_{R_2,q} \sin \phi_{R_2,q}\\
% \cos \theta_{R_2,q}
% \end{bmatrix} \nonumber\\
% \theta_{B_d}&= \arccos \left(\frac{z_U}{\| \bp_U \|_2}\right), \phi_{B_d}= \arctan 2 \left(\frac{y_U}{x_U}\right)\nonumber \\
% f_{B_d} &=\sin \theta_{B_d}\sin \phi_{B_d}  , ~g_{B_d}=\cos\theta_{B_d}. \nonumber\\
% \
% \begin{bmatrix}
% f_{U_d}\\
% g_{U_d}
% \end{bmatrix}&=
% \bM_R
% \begin{bmatrix}
% \sin\theta_{B_d} \cos \phi_{B_d}\\
% \sin\theta_{B_d} \sin \phi_{B_d}\\
% \cos \theta_{B_d}
% \end{bmatrix} \nonumber
% \end{align}
% \fi

To derive the FIM for the UE position parameters $\bxi$, we use the relation $F(\cdot)$ in \cref{eq:F}. The Jacobian matrix $\bJ \in \R^{(7+5Q) \times (5+2Q)}$ of $F$ is given in \cref{app:jaco}.
The FIM for $\bxi$ is then given by
\begin{align}\label{eq:Ibxi}
\bF_{\bxi} = \bJ\T \sum_{k=1}^{K} \bF^{(k)}_{\bmeta} \bJ \in \R^{(5+2Q) \times (5+2Q)}.
\end{align}
Accordingly, a lower bound for the MSE of the UE position is as follows:
\begin{align}
\text{MSE}(\bp_U) \ge \tr \left([\bF_{\bxi}^{-1}]_{1:3,1:3}\right).
\end{align}

\begin{Proposition} \label{pro:crb bound}
 Let
$\bxi_d  =[\bp_U\T, \opRe\{h_d\}, \opIm\{h_d\} ]$ and $ \bxi_{r,q}  =[\bp_U\T, \opRe\{ h_{r,q}\}, \opIm\{h_{r,q}\}] .$
When we only utilize parameters associated with the direct path for the UE positioning task, the error covariance matrix
\begin{align}
\bC_{{\bp}_{U}}^{(d)}=\Eb[(\hat{\bp}_{U}^{(d)} - {\bp}_U) (\hat{\bp}_{U}^{(d)} - {\bp}_U)\T] \label{eq:cov_direct}
\end{align}
satisfies the following bound,
\begin{align}
\bC_{{\bp}_{U}}^{(d)} &\succeq  \left[ ~ \widebar{\bC}_{\bxi_d}\right]_{1:3,1:3} \succeq  \left[ \left(\bJ_d\T \bF_{\bmeta_d}  \bJ_d \right)^{-1} \right]_{1:3,1:3}, \label{eq:crb direct}
\end{align}
where
$\widebar{\bC}_{\bxi_d}=\left(\bJ_d\T \widebar{\bC}_{{\bmeta}_d}^{-1}  \bJ_d \right)^{-1}$ with
$ \widebar{\bC}_{{\bmeta}_d} = \left[ \bF_{\bmeta} ^{-1}\right]_{1:7,1:7}, \bJ_d  = \frac{\partial \bmeta_d }{\partial \bxi_d\T}\in \R^{7 \times 5}$, and $\bF_{\bmeta_d} \in \C^{7\times 7}$ is the FIM of ${\bmeta_d}$.

When we only utilize parameters associated with the $q$th RIS path for the UE positioning task, the error covariance matrix
\begin{align}
\bC_{{\bp}_{U}}^{(r,q)}=\Eb[(\hat{\bp}_{U}^{(r,q)} - {\bp}_U) (\hat{\bp}_{U}^{(r,q)} - {\bp}_U)\T] \label{eq:cov_reflect}
\end{align}
satisfies the following bound,
\begin{align}
\bC_{{\bp}_{U}}^{(r,q)} &\succeq  \left[ \widebar{\bC}_{\bxi_{r,q}}\right]_{1:3,1:3} \succeq  \left[ \left(\bJ_{r,q}\T \bF_{\bmeta_{r,q}}  \bJ_{r,q} \right)^{-1} \right]_{1:3,1:3}, \label{eq:crb reflect}
\end{align}
where
$ \widebar{\bC}_{\bxi_{r,q}}=(\bJ_{r,q}\T \widebar{\bC}_{{\bmeta}_{r,q}}^{-1}  \bJ_{r,q} )^{-1}$ with
 $\widebar{\bC}_{{\bmeta}_{r,q}} = [ \bF_{\bmeta} ^{-1}]_{5q+3:7+5q,5q+3:7+5q}$,
$\bJ_{r,q}  = \frac{\partial \bmeta_{r,q}}{\partial \bxi_{r,q}\T}\in \R^{5 \times 5}$, and $\bF_{\bmeta_{r,q}} \in \C^{5\times 5}$ is the FIM of ${\bmeta}_{r,q}$.
\end{Proposition}
\begin{proof}
See \cref{app:single position}.
\end{proof}

\section{Estimation of Channel Parameters}\label{sec:estimation}

In this section, we formulate optimization problems to estimate the AoDs from the BS $(g_{B_d},v_{B_d})$, and propagation delays $\tau_d$ and $\{\tau_{r_2,q}\}_{q=1}^Q$ along the LOS path from the BS to the UE and the reflection paths from each RIS to the UE, respectively. We also estimate the AoAs at the UE $(g_{U_d},v_{U_d})$ and $(g_{U_r,q},v_{U_r,q})$ along the LOS path and reflection paths, respectively.

Because the noise $\tilde{\bN}$ in \cref{eq:R observations} is Gaussian, the \gls{MLE} of $\bmeta$ of \cref{eq:bmeta} is given by the following:
\begin{align} \label{original eta problem}
\min_{\bmeta}
\sum_{k=0}^{K-1}\| \tilde{\bR}_k - \bH_k \|_F^2.
\end{align}
However, directly solving the above problem is challenging because it is nonlinear and nonconvex in $\bmeta$. However, we note that the rank of  $\widebar{\bH}_k$ in \cref{eq:R observations} is  $Q+1$. We can leverage this low-rank property to estimate the channel parameters.

\subsection{Estimation of AoD \txp{$(g_{B_d},v_{B_d})$}{g\_Bd, v\_Bd}}

The AoD $(g_{B_d},v_{B_d})$ is for the BS-UE link given in \cref{AoD:B_d}. We discuss the estimation of $g_{B_d}$. The estimation of $v_{B_d}$ is done similarly.
We reshape $\{ \tilde{\bR}_k \}_{k=0}^{K-1}$ over the dimensions of $\ba_{\tilde{B}}(g_{B_d})$ and $\{\ba_{\tilde{B}}(g_{B_{r,q}})\}_{q=1}^Q$ as
\begin{align}
 {\bR}_{B}
 &=\brk*{\ba_{\tilde{B}}(g_{B_d}) ,\ba_{\tilde{B}}(g_{B_r,1}),\ldots,\ba_{\tilde{B}}(g_{B_r,Q})} \bQ_{B} +{\bN}_{B} \in \C^{\sqrt{N} \times \sqrt{N} DK}. \label{eq:expression RB}
\end{align}
Since $g_{B_r,q}$, for all $q$, in \cref{eq:fBrq} is known a priori as we assume that the position of the $q$th RIS is known, we only need to estimate $g_{B,d}$ from \cref{eq:expression RB} by solving
\begin{align}
&\min_{f_{B_d}} \lA{\bR}_{B} - \bA_B \bQ_{B}  \rA_F^2  \nonumber\\
&\text{subject to }\bQ_{B}  = \left(\bA_B^H \bA_B\right)^{-1}\bA_B^H {\bR}_{B}. \label{eq:problem equ}
\end{align}
where $\bA_B =\left[\bA_{B_r},\ba_{\tilde{B}}(g_{B_d}) \right]$ with $\bA_{B_r} =[ \ba_{\tilde{B}}(g_{B_r,1}),\ldots,\ba_{\tilde{B}}(g_{B_r,Q})]$. We assume $g_{B_{r,q}}$ is distinct for each $q$, which can be achieved by carefully deploying the RISs. Thus $\bA_{B_r}^H\bA_{B_r}$ is invertible and we have the following result.

\begin{Lemma}
The problem in \cref{eq:problem equ} is equivalent to
\begin{align}
\max_{g_{B_d}} \lA\widetilde{\ba}^H_{\tilde{B}}(g_{B_d}){\bR}_{B} \rA_2^2,  \label{eq:transform problem}
\end{align}
where $\widetilde{\ba}_{\tilde{B}}(g_{B_d})=\frac{\ba_{\tilde{B}}(g_{B_d}) -\bP_r \ba_{\tilde{B}}(g_{B_d})}{\| \ba_{\tilde{B}}(g_{B_d}) -\bP_r \ba_{\tilde{B}}(g_{B_d})\|_2} $ and $\bP_r =\bA_{B_r} (\bA_{B_r}^H\bA_{B_r})^{-1}\bA_{B_r}^H$.
\end{Lemma}
\begin{proof}

Note that $\bP_r \ba_{\tilde{B}}(g_{B_d})$ is the orthogonal projection onto the column space of $\bA_{B_r}$, and $\widetilde{\ba}_{\tilde{B}}(g_{B_d})$ is the residual vector of projection with normalization.
Therefore, $[\bA_{B_r},\ba_{\tilde{B}}(g_{B_d})]$ spans the same subspace as $[\bA_{B_r},\widetilde{\ba}_{\tilde{B}}(g_{B_d})]$.
For convenience, we define $\widetilde{\bA}_{B_r}$ as the Gram–Schmidt orthogonalization of columns in ${\bA}_{B_r}$.
We have
\begin{align*}
\Col([\bA_{B_r},\ba_{\tilde{B}}(g_{B_d}) ])&= \Col([\widetilde{\bA}_{B_r},\ba_{\tilde{B}}(g_{B_d}) ])\\
&= \Col([\widetilde{\bA}_{B_r},\widetilde{\ba}_{\tilde{B}}(g_{B_d}) ]),
\end{align*}
By defining $\widetilde{\bA}_B = [\widetilde{\bA}_{B_r},\widetilde{\ba}_{\tilde{B}}(g_{B_d}) ]$, one can check that $\widetilde{\bA}^H \widetilde{\bA} = \bI$.
Therefore, from the equivalence in subspaces, the residual of $\bR_B$ \gls{wrt} $\Col(\bA)$ is same as $\Col(\widetilde{\bA})$.
The objective function in \cref{eq:problem equ} is then given by
\begin{small}
\begin{align}
& \lA {\bR}_{B} - \widetilde{\bA}_B \widetilde{\bQ}_{B}  \rA_F^2 \nonumber\\
&= \tr({\bR}_{B}^H {\bR}_{B} -  {\bR}_{B}^H \widetilde{\bA}_B \widetilde{\bQ}_{B} -  \widetilde{\bQ}_{B}^H \widetilde{\bA}_B^H  {\bR}_{B} +\widetilde{\bQ}_{B}^H  \widetilde{\bA}_B^H  \widetilde{\bA}_B \widetilde{\bQ}_{B}) \nonumber\\
&= \tr({\bR}_{B}^H {\bR}_{B} -  {\bR}_{B}^H \widetilde{\bA}_B \widetilde{\bQ}_{B}) \nonumber \\
&= \tr({\bR}_{B}^H {\bR}_{B} -  {\bR}_{B}^H \widetilde{\bA}_B \widetilde{\bA}_B^H {\bR}_{B} ), \label{eq:obj equivalent}
\end{align}
\end{small}%
where $\widetilde{\bQ}_{B}=(\widetilde{\bA}_B^H \widetilde{\bA}_B)^{-1}\widetilde{\bA}_B^H {\bR}_{B}$, and  the last inequality comes from $ \widetilde{\bA}_B^H  \widetilde{\bA}_B  = \bI$.
Therefore, we have
\begin{small}
\begin{align*}
&\argmin_{g_{B_d}}\tr({\bR}_{B}^H {\bR}_{B} -  {\bR}_{B}^H \widetilde{\bA}_B \widetilde{\bA}_B^H {\bR}_{B})\\
&=\argmax_{g_{B_d}} \|  \widetilde{\bA}_B^H {\bR}_{B} \|_F^2 \\
&=\argmax_{g_{B_d}}  \|\widetilde{\ba}^H_{\tilde{B}}(g_{B_d}){\bR}_{B} \|_2^2,
\end{align*}
\end{small}%
which is exactly the problem provided in \cref{eq:transform problem}.
This concludes the proof.
\end{proof}

The variable of optimization $g_{B_d}$ in problem \cref{eq:transform problem} is scalar and various standard optimization techniques can be applied to find the optimal solution.
Suppose $\hat{g}_{B_d}$ is the optimal solution found. Let
\begin{align}\label{distinguish direct path}
\hat{\bQ}_B = \argmin_{\bQ_{B}} \|{\bR}_{B} - [\bA_{B_r},\ba_{\tilde{B}}(\hat{g}_{B_d}) ] \bQ_{B}  \|_F^2,
\end{align}
where $\hat{\bQ}_B = [ [\hat{\bQ}_B]_{r_1,:},\ldots,[\hat{\bQ}_B]_{r_Q,:},[\hat{\bQ}_B]_{d,:}]$.
Note that the values of $\|[\hat{\bQ}_B]_{d,:}]\|_2^2$ and $\|[\hat{\bQ}_B]_{r_q,:}]\|_2^2$ are related to the of energy of the direct and reflection paths.
Therefore, we can sort the paths according to the values of $\|[\hat{\bQ}_B]_{d,:}]\|_2^2$ and $\|[\hat{\bQ}_B]_{r_q,:}]\|_2^2$. We save the estimated order of the path energies as $\bs_t \in \R^{Q+1}$. This path order is utilized to distinguish the direct and reflection paths in the following subsections.

\subsection{Estimation of \txp{$\tau_d$}{tau-d} and \txp{$\{\tau_{r_2,q}\}_{q=1}^Q$}{tau-r2,q}}\label{subsec:est_delays}

We define
\begin{align}\label{eq:delay vec}
\ba_C(\tau) = [1,e^{-\iu2\pi \frac{W }{K} \tau},\ldots,e^{-\iu2\pi \frac{(K-1)W }{K} \tau}]
\end{align}
and reshape $\{ \tilde{\bR}_k \}_{k=0}^{K-1}$ over the dimensions of $\ba_C(\tau_d)$ and $\{\ba_C(\tau_{r_1,Q}+\tau_{r_2,Q})\}_{q=1}^Q$ to obtain
\begin{align}
{\bR}_D & = [\ba_C(\tau_d),\ba_C(\tau_{r_1,1}+\tau_{r_2,1}),\ldots,\ba_C(\tau_{r_1,Q}+\tau_{r_2,Q})] \bQ_D + \bN_D \in  \C^{K \times  DN }, \label{eq:expression RH}
\end{align}
where $\bQ_D \in \C^{2 \times DN}$ and $\bN_D \in \C^{K \times  DN }$. We use the multiple signal classification (MUSIC) method to estimate delays from the observations in \cref{eq:expression RH}.

Note that the column space of  ${\bR}_D$ is spanned by $\ba_C(\tau_d)$ and $\{\ba_C(\tau_{r_1,q}+\tau_{r_2,q})\}_{q=1}^Q$. Letting
$\bA_D = [\ba_C(\tau_d),\ba_C(\tau_{r_1,1}+\tau_{r_2,1}),\ldots,\ba_C(\tau_{r_1,Q}+\tau_{r_2,Q})] $, the covariance of \cref{eq:expression RH} is
\begin{align}
\bC_D  = \bA_D  \bQ_D \bQ_D ^H \bA^H +\tilde{ \sigma}^2 \bI. \label{eq:true delay cov}
\end{align}
Intuitively, when the noise level is low, the covariance matrix $\bC_D$ in \cref{eq:true delay cov} can be approximated by the covariance of the signal part, i.e., $\bA_D  \bQ_D \bQ_D ^H \bA_D^H$. This is the underlying methodology of MUSIC.
The covariance matrix in \cref{eq:true delay cov} can be estimated by using the sample correlation matrix
$\hat{\bC}_D =\bR_D \bR_D^H$.
Let $[\bw_1,\bw_2, \ldots, \bw_{ND}]$ be the eigenvectors of $\hat{\bC}_D$, where $\bw_i$ corresponds to the $i$th largest eigenvalue.
Then, letting $ \bW_{H}^c=[\bw_{Q+2},\bw_{Q+3},\ldots, \bw_{ND}]$, the estimation of the delays $\tau_d, \{\tau_{r_2,q}\}_{q=1}^Q$ is achieved by the following:
\begin{align} \label{eq:est delay music}
\text{find $Q+1$ peaks:}~~{1}/{\| \ba_C^H( \tau)
 \bW_{H}^c \|_2^2} \text{~with~} \tau \le K/W.
\end{align}

Suppose the estimated delays are $\{\hat{\tau}_i\}_{i=1}^{Q+1}$. A heuristic way to distinguish the delay for the direct path is to use the minimum delay estimated. However, this approach may result in errors when the SNR is low, as our simulation in \cref{sec:numerical} shows.
%\blue{This can be understood as follows. Assume there is one RIS, and the ground truth of time delays are $\tau_d=0.2\mu s, \tau_{r_1,1} = 0.15 \mu s, \tau_{r_2,1}= 0.25 \mu s$. If the delay of direct path is not accurate because of high fadding, we might have the estimations from \cref{eq:est delay music} are $\hat{\tau}_1=0.5\mu s$ and $\hat{\tau}_{2}=0.403 \mu s$. Then, if we distinguish the direct path based on the minimum delay, then $\hat{\tau}_2$ will be treated as the delay of direct path, and $\hat{\tau}_1$ will be treated as the reflection path. This will cause serious positioning error.}
Therefore, we use \cref{distinguish direct path} instead to assign the delays for the direct and reflection paths. Specifically, after estimating $\{\hat{\tau}_i\}_{i=1}^{Q+1}$, denoting $\hat{\bA}_D=[\ba_C(\hat{\tau}_1),\ldots,\ba_C(\hat{\tau}_{Q+1})]$, we find
\begin{align*}
\hat{\bQ}_D
&=  \argmin_{\bQ_D}\| \bR_D - \hat{\bA}_D\bQ_D\|_F^2=(\hat{\bA}_D^H\hat{\bA}_D)^{-1}\hat{\bA}_D^H \bR_D.
\end{align*}
Then, we use the path ordering $\bs_t$ from the sorting of \cref{distinguish direct path} and $\{\| [\hat{\bQ}_D]_{:,i} \|^2_2\}_{i=1}^{Q+1}$ to assign the estimated delays to path indices. Let the matched estimated delays be $\{\hat{\tau}_d,\hat{\tau}_{r_2,1},\ldots,\hat{\tau}_{r_2,Q}\}$.

\subsection{Estimation of AoAs \txp{$(g_{U_r,q},v_{U_r,q})$}{g\_Ur,v\_Urq} and \txp{$(g_{U_d},v_{U_d})$}{g\_Ud,v\_Ud}}

We present only the method to estimate $g_{U_r,q}$ and $g_{U_d}$. The same approach can be applied to the estimation of $v_{U_r,q}$ and $v_{U_d}$.
we reshape $\{ \tilde{\bR}_k \}_{k=0}^{K-1}$ over the dimension of $\ba_{\tilde{U}}(g_{U_d})$ and $\{\ba_{\tilde{U}}(g_{U_r,q})\}_{q=1}^Q$ as ${\bR}_{U} \in \C^{\sqrt{D} \times \sqrt{D} NK}$,
\begin{align}
{\bR}_{U} \!=\!\left[\ba_{\tilde{U}}(g_{U_d}),\ba_{\tilde{U}}(g_{U_r,1}),\ldots,\ba_{\tilde{U}}(g_{U_r,Q}) \right] \bQ_{U} +\tilde{\bN}_{U},
\end{align}
where $\bQ_{U}\in \C^{2 \times \sqrt{D} N K}$ and $\tilde{\bN}_{U} \in \C^{\sqrt{D} \times \sqrt{D} NK }$. Note that the signal part of the column space of $\tilde{\bR}_{U}$ is spanned by $\bA_U = \left[\ba_{\tilde{U}}(g_{U_r,1}),\ldots,\ba_{\tilde{U}}(g_{U_r,Q}),\ba_{\tilde{U}}(g_{U_d}) \right]$. Similar to \cref{subsec:est_delays}, we utilize the MUSIC method and \cref{distinguish direct path} to obtain the estimation of $\{ g_{U_r,q}\}_{q=1}^Q$ and $g_{U_d}$ as $\{\hat{g}_{U_r,q}\}_{q=1}^Q$ and $\hat{g}_{U_d}$.
Using the same technique, we can obtain the estimated $\{ v_{U_r,q}\}_{q=1}^Q$ and $v_{U_d}$ as $\{ \hat{v}_{U_r,q}\}_{q=1}^Q$ and $\hat{v}_{U_d}$, respectively.

\subsection{Estimation of \txp{$h_d$}{hd} and \txp{$\{h_r\}_{q=1}^Q$}{hq}}

Using the estimates $(\hat{\tau}_d, \hat{g}_{U_d},\hat{v}_{U_d}, \hat{g}_{B_d}, \hat{v}_{B_d})$, and  $\{\hat{\tau}_{r_2,q}, \hat{g}_{U_r,q},\hat{v}_{U_r,q}\}_{q=1}^Q$, we solve the following problem to estimate $h_d$ and $h_r$:
\begin{small}
\begin{align}
&\hat{\bh}=\argmin_{\bh} \sum_{k=0}^{K-1}\| \tilde{\bR}_k - \widebar{\bH}_k \|_F^2 \nonumber \\
&\text{subject to} ~({\tau}_d,{g}_{U_d},{v}_{U_d}, {g}_{B_d},{v}_{B_d}) = (\hat{\tau}_d,\hat{g}_{U_d},\hat{v}_{U_d}, \hat{g}_{B_d}, \hat{v}_{B_d})\nonumber \\
&~~~~~~~~~~~~ ({\tau}_{r_2,q},{g}_{U_r,q},{v}_{U_r,q}) = (\hat{\tau}_{r_2,q},\hat{g}_{U_r,q},\hat{v}_{U_r,q}), \label{estimate gain problem}
\end{align}
\end{small}%
where $\hat{\bh}$ is the estimate of $\bh=[{h_d,{h}_{r,1},\ldots,{h}_{r,Q}}]\T$.
From the formulation in \cref{eq:R observations}, we can vectorize the matrix as
\begin{align*}
  {\br}_s&=\vect([ \tilde{\bR}_0,\ldots,\tilde{\bR}_{K-1}]),
  \hat{\br}_d=\vect( [ \widebar{\bH}_{d,0},\ldots,\widebar{\bH}_{d,K-1}]),\\
 \hat{\br}_{r,q}&= \vect([ \bH_{r,0,q} \bTheta_q \bG_{0,q},\ldots,\bH_{r,K-1,q} \bTheta_q \bG_{K-1,q}]).
\end{align*}
Then, the formulated problem in \eqref{estimate gain problem} can be expressed as
\begin{align*}
\hat{\bh}=\argmin_{\bh} \| \br_s - \hat{\bR}
\bh \|_2^2,
\end{align*}
where $\hat{\bR}= [\hat{\br}_d,\hat{\br}_{r,1},\ldots,\hat{\br}_{r,Q}]$.
The solution is $\bh=(\hat{\bR}\T \hat{\bR})^{-1} \hat{\bR}^H\br_s$.

\section{UE Position Estimation} \label{sec:fusion}
In this section, we present a fusion method to infer the position of the UE from the estimated channel parameters.

\subsection{Fusion via Linear Combination} \label{sec:linear fusion}
Recall that the error covariance matrices of $\hat{\bp}_{U}^{(d)}$ and $\{\hat{\bp}_{U}^{(r,q)}\}_{q=1}^Q$ are given by \cref{eq:cov_direct} and \cref{eq:cov_reflect}, respectively.
The following lemma presents the proposed fusion method based on the error covariance matrices.

\begin{Lemma}\label{lem:blue}
Assume that the estimations of $\hat{\bp}_{U}^{(d)}$ and $\{\hat{\bp}_{U}^{(r,q)}\}_{q=1}^Q$ are based on independent measurements, and the covariance matrices are in \cref{eq:cov_direct} and \cref{eq:cov_reflect}. Then, the optimal linear combination of $\hat{\bp}_{U}^{(d)}$ and $\{\hat{\bp}_{U}^{(r,q)}\}_{q=1}^Q$ is given by
\begin{align} \label{eq:linear combination}
\hat{\bp}_{U} = {\bC}_{\bp_U}(   ( \bC_{{\bp}_{U}}^{(d)} )^{-1} \hat{\bp}_{U_d}  + \sum_{q=1}^{Q}(\bC_{{\bp}_{U}}^{(r,q)})^{-1} \hat{\bp}_{U_r,q}),
\end{align}
where ${\bC}_{\bp_U} = ( ( \bC_{{\bp}_{U}}^{(d)} )^{-1}+  \sum_{q=1}^{Q} (\bC_{{\bp}_{U}}^{(r,q)})^{-1} )^{-1}$.
\end{Lemma}
\begin{proof}
To obtain a linear combination of $\hat{\bp}_{U}^{(d)}$ and $\{\hat{\bp}_{U}^{(r,q)}\}_{q=1}^Q$ as the estimate of the UE position, we let the expression of the estimation of UE position be
\begin{align}
\hat{\bp}_U &= \bA_d \hat{\bp}_{U}^{(d)} +\sum_{q=1}^{Q}\bB_q \hat{\bp}_{U}^{(r,q)}, \label{eq:pos user expression}
\end{align}
where $\bA_d \in \C^{3\times 3}$ and $\bB_q\in \C^{3\times 3}, \forall q$.
In order to obtain an unbiased estimator, it must have $\bA_d +\sum_{q=1}^{Q}\bB_q = \bI$.
To minimize the MSE of $\hat{\bp}_U$, we need to solve the following problem:
\begin{align}
\min _{\bA_d, \{\bB_q \}_{q=1}^Q} &  \tr(\Eb \left[(\hat{\bp}_U - \bp_U) (\hat{\bp}_U - \bp_U)\T\right]) \nn
\text{subject to~} & \bA_d +  \sum_{q=1}^{Q}  \bB_q = \bI. \label{eq:problem blue}
\end{align}
Substituting the expression of \cref{eq:pos user expression} and taking first order derivative of the objective function in \cref{eq:problem blue} give
\begin{align*}
\bA_d ={\bC}_{\bp_U} ({\bC}_{{\bp}_{U}}^{(d)})^{-1},~\bB_q = {\bC}_{\bp_U} (\bC_{{\bp}_{U}}^{(r,q)})^{-1}.
\end{align*}
This concludes the proof.
\end{proof}
\begin{Remark} \label{remark:bound of Cpu}
Since we assume independence among the UE estimations from these paths, according to \cref{pro:crb bound}, we have the following bounds
\begin{align}
{\bC}_{{\bp}_{U}}^{(d)}& \succeq {\left[ \left(\bJ_d\T \bF_{\hat{\bmeta}_d}  \bJ_d \right)^{-1} \right]_{1:3,1:3}}=\tilde{\bC}_{{\bp}_{U}}^{(d)}, \label{eq:def_Cd}\\
{\bC}_{{\bp}_{U}}^{(r,q)}& \succeq {\left[ \left(\bJ_{r,q}\T \bF_{\hat{\bmeta}_{r,q}}  \bJ_{r,q} \right)^{-1} \right]_{1:3,1:3}}=\tilde{\bC}_{{\bp}_{U}}^{(r,q)},\label{eq:def_Cr}
\end{align}
where we denote the bounds as $\tilde{\bC}_{{\bp}_{U}}^{(d)}$ and $\tilde{\bC}_{{\bp}_{U}}^{(r,q)}$. We define $\tilde{\bC}_{\bp_U} = ( ( \tilde{\bC}_{{\bp}_{U}}^{(d)} )^{-1}+  \sum_{q=1}^{Q} (\tilde{\bC}_{{\bp}_{U}}^{(r,q)})^{-1} )^{-1}$. Therefore, when the exact error covariances in \cref{eq:linear combination} are not available, we can employ the lower bounds in \cref{eq:def_Cd} and \cref{eq:def_Cr},
\begin{align}\label{eq:linear combination bound}
\hat{\bp}_{U} = \tilde{\bC}_{\bp_U} (  (\tilde{\bC}^{(d)}_{\bp_U})^{-1} \hat{\bp}_{U}^{(d)}  +   \sum_{q=1}^{Q} (\tilde{\bC}^{(r,q)}_{\bp_U})^{-1}\hat{\bp}_{U}^{(r,q)}).
\end{align}
\end{Remark}

\subsection{Asymptotic MLE} \label{sect:asy MLE}
We now show that the proposed linear combination in \cref{eq:linear combination bound} is approximately the \gls{MLE} in the asymptotic regime of large sample size.
We first introduce the extended invariance principle (EXIP), which is asymptotically equivalent to the \gls{MLE}.
Then we show that the proposed linear combination is approximately the optimal solution of EXIP.

\begin{Theorem}[EXIP theorem\cite{stoica1989reparametrization}]\label{thm:EXIP}
Suppose the loss function for estimating the parameters $\bxi $ is given by $L(\by; \bxi)$, where $\by\in \R^{Z\times 1}$ are observations. Suppose there exists a function $\bmeta = F(\bxi)$ with loss function
$L(\by; \bmeta) = L(\by; F(\bxi) )=L(\by; \bxi)$.
The estimation of $\bxi$ and $\bmeta$ are given by
\begin{align*}
\hat{\bxi} &= \argmin L(\by; \bxi),\ \hat{\bmeta} =\argmin L(\by; \bmeta).
\end{align*}
If
$
\lim_{N \rightarrow \infty}\hat{\bmeta} = \lim_{N \rightarrow \infty}F(\hat{\bxi}),
$
then
\begin{align} \label{eq:EXIP form}
 \breve{\bxi} = \argmin_{\bxi} [\hat{\bmeta}-F(\bxi)]\T\bW [\hat{\bmeta}-F(\bxi)]
\end{align}
is asymptotically equivalent to $\hat{\bxi}$ as $Z\to\infty$, where $\bW $ is
\begin{align*}
\bW = \evalat*{\Eb\left[\frac{\partial^2 L(\by; \bmeta)}{\partial \bmeta \bmeta\T}\right]}_{\bmeta = \hat{\bmeta}}.
\end{align*}
\end{Theorem}

In \cref{eq:F}, the UE position parameters are related to the channel parameters $\bmeta$ via a function $F(\cdot)$. We can apply the EXIP approach to obtain the UE position estimate from the channel parameter estimates. Specifically, from the estimator $\hat{\bmeta}$ in \cref{sec:estimation}, applying \cref{thm:EXIP}, we can solve the following weighted least squares problem
\begin{align}
\breve{\bxi} = \argmin_{\bxi} [\hat{\bmeta}-F(\bxi)]\T\bW [\hat{\bmeta}-F(\bxi)], \label{eq:UE position estimation}
\end{align}
where the weight matrix $\bW $ is given by
\begin{align*}
 \bW = \evalat*{\Eb\left[\frac{\partial^2 L(\{\tilde{\bR}_k\}_{k=1}^K; \bmeta)}{\partial \bmeta \bmeta\T}\right]}_{\bmeta = \hat{\bmeta}}=\bF_{\hat{\bmeta}},
\end{align*}
and the loss function $L(\{\tilde{\bR}_k\}_{k=1}^K; \bmeta) = \sum_{k=1}^{K}\| \tilde{\bR}_k - \widebar{\bH}_k \|_F^2$ (see \cref{original eta problem}).
Note that the inference model in \cite{Lin2021Channel,Jing2020Positioning} with geometric mapping is equivalent to letting $\bW = \bI$ in \cref{eq:UE position estimation}.
In particular, the gradient-based method can be utilized to find the optimum in \cref{eq:UE position estimation}, which is, however, sensitive to the initialization.
In what follows, we will show that \cref{eq:linear combination bound} is the approximate solution of the problem \cref{eq:UE position estimation}.

Let $F_d$ and $F_{r,q}$, $q=1,\ldots,Q$, be functions such that $\bmeta_d=F_d(\bxi_{d})$ and $\bmeta_{r,q}=F_{r,q}(\bxi_{r,q})$. The objective function in \cref{eq:UE position estimation} is approximate to the following,
\begin{align}
&(\hat{\bmeta}-F(\bxi))\T\bW (\hat{\bmeta}-F(\bxi)) \nonumber \\
 &\approx
  \begin{bmatrix}
   F_d(\hat{\bxi}_d)-F_d(\bxi_d) \\
  F_{r,1}(\hat{\bxi}_{r,1})-F_{r,1}(\bxi_{r,1}) \\
  \vdots \\
  F_{r,1}(\hat{\bxi}_{r,Q})-F_{r,Q}(\bxi_{r,Q}) \\
 \end{bmatrix}\T
\bF_{\hat{\bmeta}}
  \begin{bmatrix}
   F_d(\hat{\bxi}_d)-F_d(\bxi_d) \\
  F_{r,1}(\hat{\bxi}_{r,1})-F_{r,1}(\bxi_{r,1}) \\
  \vdots \\
  F_{r,1}(\hat{\bxi}_{r,Q})-F_{r,Q}(\bxi_{r,Q}) \\
 \end{bmatrix}\nonumber \\
  &\approx
   \begin{bmatrix}
   \bJ_d(\hat{\bxi}_d-\bxi_d) \\
 \bJ_{r,1}(\hat{\bxi}_{r,1}-\bxi_{r,1}) \\
  \vdots \\
  \bJ_{r,1}(\hat{\bxi}_{r,Q}-\bxi_{r,Q}) \\
 \end{bmatrix}\T
\bF_{\hat{\bmeta}}
  \begin{bmatrix}
   \bJ_d(\hat{\bxi}_d-\bxi_d) \\
  \bJ_{r,1}(\hat{\bxi}_{r,1}-\bxi_{r,1}) \\
  \vdots \\
  \bJ_{r,1}(\hat{\bxi}_{r,Q}-\bxi_{r,Q}) %
 \end{bmatrix},\label{eq:partial decom EXIP} %
\end{align}%
where $\hat{\bxi}_d$ is inferred from $\hat{\bmeta}_{d}$, and  $\hat{\bxi}_{r,q}$ is inferred from $\hat{\bmeta}_{r,q}$.
Further details are in \cref{sect: cal pU}.
The first approximation is from that $F_d(\hat{\bxi}_d)\approx \hat{\bmeta}_d$ and $F_{r,q}(\hat{\bxi}_{r,q})\approx \hat{\bmeta}_{r,q}$.
The second approximation in \cref{eq:partial decom EXIP} holds from the Taylor series expansion.
Letting the first order derivative of \cref{eq:partial decom EXIP} be zero gives
\begin{align}
\hat{\bxi} ={( \bJ\T\bF_{\hat{\bmeta}}\bJ)^{-1}
\bJ\T
\bF_{\hat{\bmeta}}}
\begin{bmatrix}
   \bJ_d \hat{\bxi}_d \\
  \bJ_{r,1}\hat{\bxi}_{r,1} \\
  \vdots \\
  \bJ_{r,Q}\hat{\bxi}_{r,Q}
 \end{bmatrix}.   \label{eq:approximate EXIP}
\end{align}
Therefore, the solution in \cref{eq:approximate EXIP} is the approximate solution of \cref{eq:UE position estimation}, which is asymptotically \gls{MLE} in the large-sample regime.

The following proposition shows that the optimal linear combination in \cref{eq:linear combination bound} is equivalent to the solution in \cref{eq:approximate EXIP} when the paths are independent.

\begin{Proposition} \label{lem:equivalence}
Suppose the paths are independent, in other words, $\bF_{\hat{\bmeta}}$ has the following form
\begin{align}
\bF_{\hat{\bmeta}} =
\begin{bmatrix}
  \bF_{\hat{\bmeta}_d} & \boldsymbol{0}& \cdots & \boldsymbol{0}\\
   \boldsymbol{0} & \bF_{\hat{\bmeta}_{r,1}} & \cdots & \boldsymbol{0} \\
   \vdots & \vdots& \ddots& \vdots\\
      \boldsymbol{0} &    \boldsymbol{0} & \cdots & \bF_{\hat{\bmeta}_{r,Q}}
\end{bmatrix}.
\end{align}
Then, the solution in \cref{eq:approximate EXIP} is equivalent to \cref{eq:linear combination bound}, that is
$[\hat{\bxi}]_{1:3} =\tilde{\bC}_{\bp_U} (  (\tilde{\bC}^{(d)}_{\bp_U})^{-1} \hat{\bp}_{U}^{(d)}  +   \sum_{q=1}^{Q} (\tilde{\bC}^{(r,q)}_{\bp_U})^{-1}\hat{\bp}_{U}^{(r,q)}).$

\end{Proposition}
\begin{proof}
See \cref{app:equivalence}.
\end{proof}

In summary, we have shown that the optimal linear combination in \cref{eq:linear combination bound} is approximate to the optimal solution of EXIP method when the paths are independent.
Therefore, it is approximately equivalent to the MLE in large-sample region .

\subsection{Estimation of \txp{$\hat{\bp}_{U}^{(d)}$}{p\_Ud} and  \txp{$\{\hat{\bp}_{U}^{(r,q)}\}_{q=1}^Q$}{p\_Uq}}  \label{sect: cal pU}

\subsubsection{Inferring \txp{$\hat{\bp}_{U}^{(d)}$}{p\_Ud} from \txp{$\hat{\bgamma}_{d}$}{gamma\_d}}

We first focus on the direct path, and discuss how to obtain the refined channel parameters associated with the direct path. Define ${\boldf} = [{g}_{U_d},{v}_{U_d},{g}_{B_d},{v}_{B_d}]\T$,
and $\bz = [f_{B,d},{g}_{B_d},{v}_{B_d}]\T$ with $f_{B,d}=\sin\theta_{B_d} \cos \phi_{B_d}$.
From \cref{AoD:B_d,AoD:U_d}, we have
$
{\boldf}=
[
      \bM_R;
      \boldsymbol{0}, ~\bI
] \bz= \tilde{\bA}\bz.
$
However, the estimation result
 $\hat{\boldf}  = [\hat{g}_{U_d},\hat{v}_{U_d},\hat{g}_{B_d},\hat{v}_{B_d}]\T$ may not satisfy the above relation due to corruption of noise.
We employ the weighted least squares method by solving
\begin{align} \label{eq:direct Pu}
\min_{\bz} (\hat{\boldf} - \tilde{\bA} \bz)\T \widebar{\bC}_{\hat{\boldf}}^{-1}  (\hat{\boldf} - \tilde{\bA} \bz),\ \text{subject to} ~\|  \bz\|_2^2 = 1,
\end{align}
where
$\widebar{\bC}_{\hat{\boldf}} =[\bF_{\hat{\bmeta}}^{-1}]_{2:5,2:5}$.
If we ignore the constraint, the solution is given by
$\hat{\bz} =(\tilde{\bA} \T \widebar{\bC}_{\hat{\boldf}}^{-1} \tilde{\bA})^{-1} \tilde{\bA} \T \widebar{\bC}_{\hat{\boldf}}^{-1} \hat{\boldf}$.
Then we project this solution to the feasible region of the problem in \cref{eq:direct Pu}.
The estimation of $(\hat{d},\hat{\theta}_{B_d},\hat{\phi}_{B_d})$ is given by
\begin{align}
\hat{d} = c\hat{\tau}_d,
\hat{\theta}_{B,d} = \arccos \hat{v}_{B,d},
\hat{\phi}_{B,d} = \arcsin ({\hat{g}_{B,d}}/{\sin \hat{\theta}_{B,d}}). \nonumber
\end{align}
The estimated UE position $\hat{\bp}_{U}^{(d)}$ from the direct path is then given by
\begin{align}
\begin{cases}
\hat{z}_{U_d}= \hat{d}   \cos \hat{\theta}_{B_d} =  \hat{d}   \hat{v}_{B,d} \nonumber\\
\hat{x}_{U_d}= \hat{d} \sin \hat{\theta}_{B_d}  \cos \hat{\phi}_{B_d}   = \hat{d} \sqrt{1-\hat{g}_{B,d}^2 - \hat{v}_{B,d}^2} \nonumber\\
\hat{y}_{U_d} = \hat{d} \sin \hat{\theta}_{B_d}  \sin \hat{\phi}_{B_d} =  \hat{d} \hat{g}_{B,d}.
\end{cases}
\end{align}

\subsubsection{Inferring \txp{$\hat{\bp}_{U}^{(r,q)}$}{p\_Uq} from \txp{$\hat{\bgamma}_{r,q}$}{gamma\_r}}\label{sect: cal pUr}
For the $q$th reflection path, we define $g_{R_2,q}  = \sin \theta_{R_2,q} \sin \phi_{R_2,q}$ and  ${f}_{U_r,q} = -\sqrt{1-{g}_{U_r,q}^2-{v}_{U_r,q}^2}$.
Based on the relations in \cref{eq:expression barMR} and \cref{AoA:U_r_q}, we have
\begin{align}
\begin{bmatrix}
f_{R_2,q}\\
g_{R_2,q}\\
v_{R_2,q}
\end{bmatrix} =
\begin{bmatrix}
\sin\theta_{R_2,q} \cos \phi_{R_2,q}\\
\sin\theta_{R_2,q} \sin \phi_{R_2,q}\\
\cos \theta_{R_2,q}
\end{bmatrix}
=
\widebar{\bM}_R^{-1}\begin{bmatrix}
f_{U_r,q}\\
g_{U_r,q}\\
v_{U_r,q}
\end{bmatrix}.
\end{align}
Therefore,
$[\hat{f}_{R_2,q},\hat{g}_{R_2,q},\hat{v}_{R_2,q}]\T
=\widebar{\bM}_R^{-1}[\hat{f}_{U_r,q},\hat{g}_{U_r,q},\hat{v}_{U_r,q}]\T$.
The estimation of $\{\hat{d}_{2,q},\hat{\theta}_{R_2,q} ,\hat{\phi}_{R_2,q} \}_{q=1}^Q$ is given by
\begin{align*}
\hat{d}_{2,q} \!\!= \! c\hat{\tau}_{r_2,q},\hat{\theta}_{R_2,q}\! \!=\! \arccos \hat{v}_{R_2,q},
\hat{\phi}_{R_2,q}  \!\!= \!\arctan 2(\hat{g}_{R_2,q},\! \hat{f}_{R_2,q}). \nonumber \nonumber
\end{align*}
Then, the UE position can be estimated as
\begin{align}
\begin{cases}
\hat{z}_{U_r,q} - {z}_{R,q}= \hat{d}_{2,q} \cos \hat{\theta}_{R_2,q}  = \hat{d}_{2,q}  \hat{v}_{R_2,q},\nonumber\\
\hat{x}_{U_r,q}-{x}_{R,q}= \hat{d}_{2,q} \sin \hat{\theta}_{R_2,q}  \cos \hat{\phi}_{R_2,q}   =  \hat{d}_{2,q} \hat{f}_{R_2,q}, \nonumber\\
\hat{y}_{U_r,q}-{y}_{R,q} = \hat{d}_{2,q}\sin \hat{\theta}_{R_2,q}  \sin \hat{\phi}_{R_2,q}= \hat{d}_{2,q} \hat{g}_{R_2,q}.   \nonumber
\end{cases}
\end{align}

\section{Discussions} \label{sec:discussions}
In this section, we propose methods to optimize the phase shifts of a RIS for the purpose of positioning a UE. We also discuss the extension of our proposed framework to the multi-BS and multi-UE scenarios.

\subsection{Design of RIS Phase Shifts}

\begin{figure}[!htpb]
\center
\includegraphics[width=0.6\textwidth]{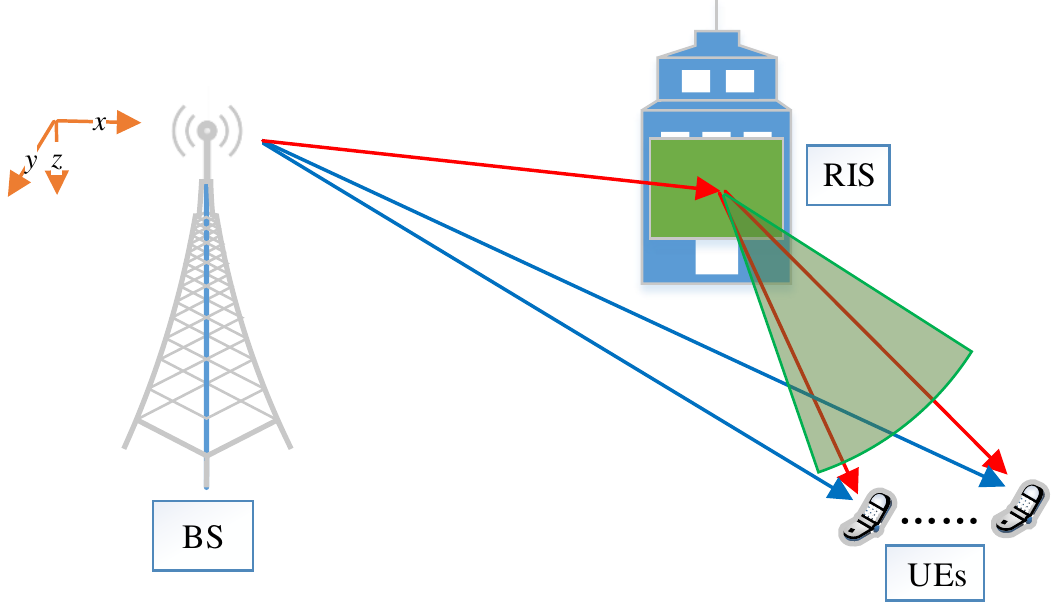}
\caption{Positioning of multiple UEs with the aid of a RIS. } \label{figure frame 1}
\end{figure}
% \iffalse
% \begin{figure}
% \center
% \includegraphics[width=3.0 in]{figures/frame_2.pdf}
% \caption{Feedback.} \label{figure frame 2}
% \end{figure}
% \fi

We consider the RIS-aided positioning scenario in \cref{figure frame 1}, where the RIS aims to serve mutiple UEs. Specifically,
the phase shifts of a RIS are designed to serve the UEs with elevation angles in the range $[\theta_l, \theta_u]$ and azimuth angles in $[\phi_l, \phi_u]$. For example, the LOS between the BS and UEs within this region of interest may be blocked with high probability. The phase shifts of the RIS are designed to aid these UEs.
%On the other hand, in the feedback scenario of \cref{figure frame 2}, the phase shifts of the RIS is designed based on the feedback from a UE.

Recall that the gain of the reflection path is proportional to $| \ba_R^H(f_{R_2,q},v_{R_2,q}) \bTheta_q  \ba_R(f_{R_1,q},v_{R_1,q})| $.
Since the quantities $f_{R_1,q}$ and $v_{R_1,q}$ are unknown a priori, we make an unbiased design $\bTheta_q$ based on the served UEs in the following.

For convenience, we combine $ \bTheta_q  \ba_R(f_{R_1,q},v_{R_1,q})$ as one variable $ \tilde{\btheta}_q \in \C^{M \times 1}$. If the UEs are uniformly distributed  in elevation range $[\theta_l, \theta_u]$ and azimuth range $[\phi_l, \phi_u]$, then we consider the following optimization problem:
\begin{align}
&\max_{ \tilde{\btheta}_q } \Eb\brk*{ | (\ba_{\tilde{R}}(f) \otimes  \ba_{\tilde{R}}(v) )^H \tilde{\btheta}_q  |^2} \nonumber \\
&\text{subject to~} f = \sin\theta \cos \phi, ~v = \cos\theta, \nonumber \\
 &~~~~~~~~~~~~ \theta \sim U [{\theta}_{q,l},{\theta}_{q,u}    ], ~\phi \sim U [{\phi}_{q,l},{\phi}_{q,u}    ]. \label{eq: phase expectation}
\end{align}
Since directly solving \cref{eq: phase expectation} is challenging, we define a matrix
 $\bD_A \in \C^{M\times Z}$ with its column having the form of $ \ba_{\tilde{R}}(f) \otimes  \ba_{\tilde{R}}(v) $, where  $(f,v)$ is chosen in a discretized range. Therefore, we have the following approximation
\begin{align}
\Eb\brk*{ | (\ba_{\tilde{R}}(f) \otimes  \ba_{\tilde{R}}(v) )^H \tilde{\btheta}_q  |^2} \approx \frac{1}{Z} \| \bD_A^H \tilde{\btheta}_{q}   \|_2^2.
\end{align}
 Then, we can reformulate the  problem in \cref{eq: phase expectation} as
\begin{align}
\hat{\tilde{\btheta}}_q = \argmax_{\tilde{\btheta}_{q} }\| \bD_A^H \tilde{\btheta}_{q} \|_2^2. \label{eq: reform RIS phase}
\end{align}
If there is no constraint for $\hat{\tilde{\btheta}}_q $, the solution is the dominant left singular vector of $\bD_A$.
In order to satisfy the constraint imposed on $\tilde{\btheta}_{q}$, we let $\hat{\tilde{\btheta}}_q $ be the complex angle of the dominant left singular vector of $\bA$.
The design of phase shifts of $q$th RIS $\bTheta_q$ is then given by  $\diag(\btheta_q)$ with
\begin{align}
\angle\btheta_q = -\angle\ba_R(f_{R_1,q},v_{R_1,q}) + \angle \hat{\tilde{\btheta}}_q.
\end{align}

% \iffalse
% \subsubsection{Feedback Scenario}
% When the position of UE is feedback to the BS, thus the $\hat{f}_{R_2,q},\hat{g}_{R_2,q}$ can be estimated. Therefore, the design of phase shifts of RIS can be given by solving
% \begin{align*}
% \max_{ \bTheta_q } | \ba_R^H(\hat{f}_{R_2,q},\hat{g}_{R_2,q}) \bTheta_q  \ba_R(f_{R_1,q},g_{R_1,q})|.
% \end{align*}
% The design of $\bTheta_q$ is given by
% \begin{align}
% \angle\btheta_q = -\angle\ba_R(f_{R_1,q},g_{R_1,q}) + \angle \ba_R(\hat{f}_{R_2,q},\hat{g}_{R_2,q}) .
% \end{align}
% \begin{Remark} \label{remark RIS elements}
% When the elements of RIS are optimally designed, the path loss of reflection path
% $\beta_r \propto M^2$.
% Accordingly, the estimation error from the reflection path
% $\text{MSE}(\hat{\bp}_r) \propto  \frac{1}{M^2}$.
% \end{Remark}
% \fi

\subsection{Extension to Multiple UEs} \label{section multi UEs}
%<*tag:QQ2>
Assume there are $L$ users. Let the true position of UE $l$ be $\bp_{U,l} = \bp_{U,1}+\Delta_l$ with $\Delta_l$ being the relative position \gls{wrt} the reference UE 1 with position $\bp_{U,1}$. We assume that the relative positions of the UEs are known through inter-UE measurements and message exchanges \cite{wymeersch2009cooperative,xu2015distributed,conti2012network}.
Assume each UE $l$ first estimates its position independently as $\hat{\bp}_{U,l}$ with error covariance ${\bC}_{\bp_{U,l}}$.
To leverage on inherent correlations among the UEs,  we employ the linear combination of estimates as described in \cref{sec:linear fusion},
\begin{align}
\hat{\bp}_{U,1}&=\sum_{l=1}^{L} \bA_{l} (\hat{\bp}_{U,l}-\Delta_l),
\hat{\bp}_{U,l} = \hat{\bp}_{U,1}+\Delta_l, \label{eq:muti UE problem}
\end{align}
where $\bA_{l}\in \C^{3 \times 3}$ is the combining matrix.
By using the similar statements as \cref{lem:blue}, we can
minimize the MSE of the estimates, i.e., $\sum_{l=1}^{L}\tr(\Eb[ (\hat{\bp}_{U,l}- {\bp}_{U,l})(\hat{\bp}_{U,l}-{\bp}_{U,l})^H])$, and obtain the expression of $\bA_l$ in \cref{eq:muti UE problem} as follows,
\begin{align*}
\bA_l&=\left(\sum_{l=1}^{L} {\bC}_{\bp_{U,l}}^{-1}\right)^{-1} {\bC}_{\bp_{U,l}}^{-1}.
\end{align*}

Therefore, the estimated positions of UEs are given by
\begin{align}
\hat{\bp}_{U,1}&=\left(\sum_{l=1}^{L} {\bC}_{\bp_{U,l}}^{-1}\right)^{-1} \sum_{l=1}^{L}{\bC}_{\bp_{U,l}}^{-1}(\hat{\bp}_{U,l}-\Delta_l),\label{eq:muti UE solution}\\
\hat{\bp}_{U,l} &= \hat{\bp}_{U,1}+\Delta_l. \nonumber
\end{align}
After the fusion, the resulting error covariance of each UE is given by $(\sum_{l=1}^{L} {\bC}_{\bp_{U,l}}^{-1})^{-1}, \forall l$.
This result can be utilized for the case where there are multiple BSs, which we discuss in the following subsection.
In particular, when the error covariance ${\bC}_{\bp_{U,l}}$ in \cref{eq:muti UE solution} is not available, we can employ the lower bound as an alternative, which can still achieve near optimal performance as we analyzed in \cref{sec:linear fusion} and \cref{sect:asy MLE}.

\subsection{Extension to Multiple BSs} \label{sect:multi BSs}
We now consider the case where there are $P$ BSs. The received signal of the $k$th subcarrier at the UE is given by
\begin{align}
\bR_k =\sum_{i=1}^{P} \bH_{i,k} \bX_{i} + \bN_k, \label{eq:compact signal muti BS}
\end{align}
where $\bX_i \in \C^{D \times T}$.
As in the case of a single BS, we assume $\bX_i \bX_i^H ={T}/{D} \bI$. Here, we further assume that $\bX_i \bX_j^H  = \boldsymbol{0},\ \forall i\neq j$.
Thus, right multiplying both sides of \cref{eq:compact signal muti BS} with ${D}/{T}\bX_i^H $ yields
\begin{align}
\frac{D}{T}\bR_k \bX_{i}^H &=\left(\sum_{i=1}^{P} \frac{D}{T}\bH_{i,k} \bX_{i}\right)\bX_{i}^H + \frac{D}{T}\bN_k\bX_{i}^H \nonumber \\
 &= \bH_{i,k} + \frac{D}{T} \bN_k\bX_{i}^H.
\end{align}
From the signal transmitted by the $i$th BS, we can estimate the UE position by using the proposed method for a single BS, i.e., $\hat{\bp}_{U}^{(i)}$. Let the error covariance be ${\bC}_{{\bp}_{U}^{(i)}}$.
Then, the estimate from different BSs can be fused by linear combination as
\begin{align} \label{eq:fusion muti BS problem}
\hat{\bp}_{U} =\sum_{i=1}^{P}\bB_i\hat{\bp}_{U}^{(i)},
\end{align}
where $\bB_i \in \C^{3 \times 3}$ is the combining matrix.
Similarly, by minimizing the MSE of $\hat{\bp}_{U}$ in \cref{eq:fusion muti BS problem}, i.e., $\tr(\Eb[ \hat{\bp}_{U}- {\bp}_{U})(\hat{\bp}_{U}-{\bp}_{U})^H])$,
the final estimate from different BSs is given by
\begin{align*}
\hat{\bp}_{U} = \left( \sum_{i=1}^{P}{\bC}_{{\bp}_{U}^{(i)}} ^{-1}\right)^{-1}\sum_{i=1}^{P}{\bC}_{{\bp}_{U}^{(i)}}^{-1}\hat{\bp}_{U}^{(i)}.
\end{align*}

\section{Numerical Results}\label{sec:numerical}
In this section, we evaluate the proposed RIS-aided positioning method. We verify the UE positioning accuracy achieved by comparing to the CRB under varying the noise level. We also verify the channel parameter estimation accuracy. Numerical experiments are also conducted to provide insights into the impact of the direct path loss exponent on the UE positioning accuracy. Finally, we present the simulation results for the multi-UE and multi-BS scenarios.

In the simulations, we utilize the root mean-square error, RMSE $=\sqrt{\Eb[\|\bp_U-\hat{\bp}_U\|_2^2]}$, to measure the positioning accuracy, where $\hat{\bp}_U$ is the estimated UE position.
Throughout our experiments, we use the parameter settings in \cref{tab:simulation parameter}.

\begin{table}
  \footnotesize
  \centering
  \caption{Simulation Parameters}
  \begin{tabular}{|c|c|}\hline
    \textbf{Parameter}         & \textbf{Value}      \\ \hline
    Number of BS antennas      & $N=100$             \\ \hline
    Number of UE antennas      & $D=64$              \\ \hline
    RIS size                   & $M=400$             \\ \hline
    Transmission bandwidth     & $W = 100\text{MHz}$ \\ \hline
    Carrier frequency          & $f_c=30\text{GHz}$  \\ \hline
    Number of OFDM subcarriers & $K=32$              \\ \hline
    Rician factor              & $K_d=100$           \\ \hline
    Number of time slots       & $T=6\times 10^5$    \\  \hline
  \end{tabular}
  \label{tab:simulation parameter}
\end{table}

\subsection{UE Positioning Accuracy}

\begin{figure}[!htb]
\centering
\includegraphics[width=.6\textwidth]{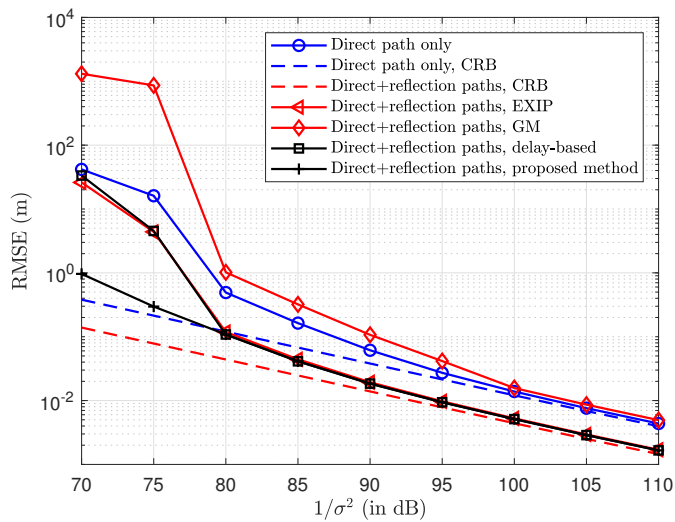}
\caption{UE position estimation RMSE versus the noise level.}  \label{fig:fig position error}
\end{figure}

In this simulation, we evaluate the UE positioning accuracy of the proposed RIS-aided positioning method with a single BS and UE. All position coordinates are measured in meters with the BS at the origin as illustrated in \cref{fig:system}.
The UE position is $\bp_U=[50,10,20]$, where $z_U=20$ means the UE height of  $20$m from the BS. One RIS is at position $\bp_{R,1}=[30,-5,2]$, where $z_{R,1}=2$ means the RIS height of $2$m from the BS.
The path loss exponent for the direct path is $L_d=4.5$, and the path loss exponent for reflection path is $L_r=2$.
We evaluate the positioning accuracy following methods:
\begin{itemize}
  \item The proposed positioning method that distinguishes the direct path based on the path energy, labeled as ``Direct+reflection paths, proposed method''.
  \item The proposed positioning method that distinguishes the direct path based on the estimated delay, labeled as ``Direct+reflection paths, delay-based''.
  \item The positioning method with EXIP \cite{stoica1989reparametrization}, labeled as ``Direct+reflection paths, EXIP''.
  \item The positioning method with geometric mapping \cite{Lin2021Channel,Jing2020Positioning}, labeled as ``Direct+reflection paths, GM''.
  \item The positioning method that utilizes only the direct path, labeled as ``Direct path only''.
\end{itemize}
We observe from \cref{fig:fig position error} that our proposed method outperforms the benchmark approaches, with RMSE close to the CRB when the SNR is high. The result also verifies that distinguishing the paths based on the path energy provides better performance compared to the delay-based approach.
We also observe that fusing the estimates from the direct and reflection paths achieves a better accuracy than using only the direct path, which validates the effectiveness of the RIS.

\subsection{Channel Parameter Estimation Accuracy}

\begin{figure}[!tb]
\centering
\includegraphics[width=.6\textwidth]{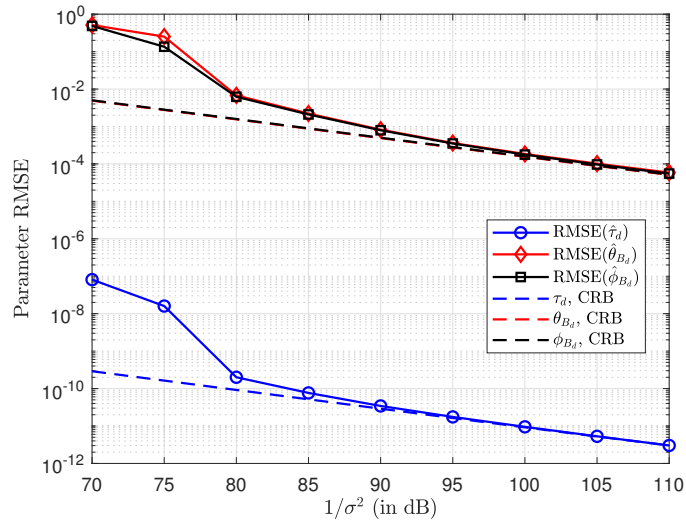}
\caption{Channel parameter estimation RMSE of the direct path versus the noise level.}  \label{fig:direct param}
\end{figure}
\begin{figure}[!tb]
\centering
\includegraphics[width=.6\textwidth]{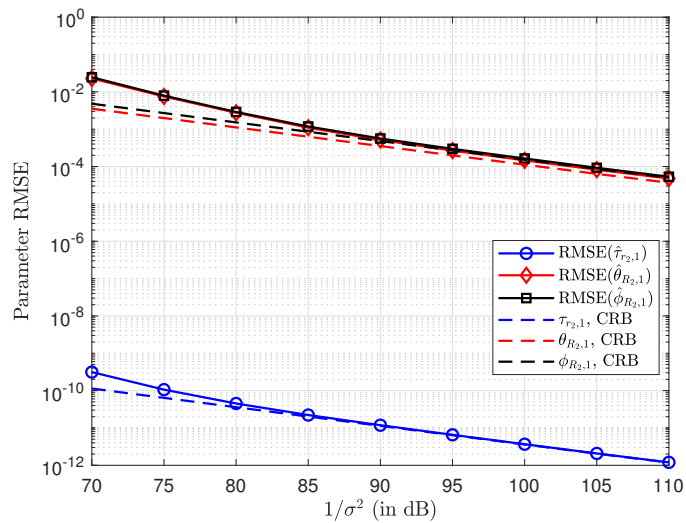}
\centering
\caption{Channel parameter estimation RMSE of the reflection path versus the noise level.}  \label{fig:ref param}
\end{figure}

In \cref{fig:direct param,fig:ref param}, we evaluate the RMSE of the channel parameters, i.e., $(\tau_d, \theta_{B_d},\phi_{B_d})$ and $(\tau_{r_2,1}, \theta_{R_2,1},\phi_{R_2,1})$, by using proposed method.
The CRBs of the estimators are also plotted as the benchmark.
The simulation settings are the same as those in \cref{fig:fig position error}. As we expect in \cref{fig:direct param,fig:ref param}, the RMSEs of the estimated parameters are all close to their CRBs, which validates the effectiveness of the proposed method.

\subsection{Direct Path Loss Exponent}

\begin{figure}[!htb]
\centering
\includegraphics[width=.6\textwidth]{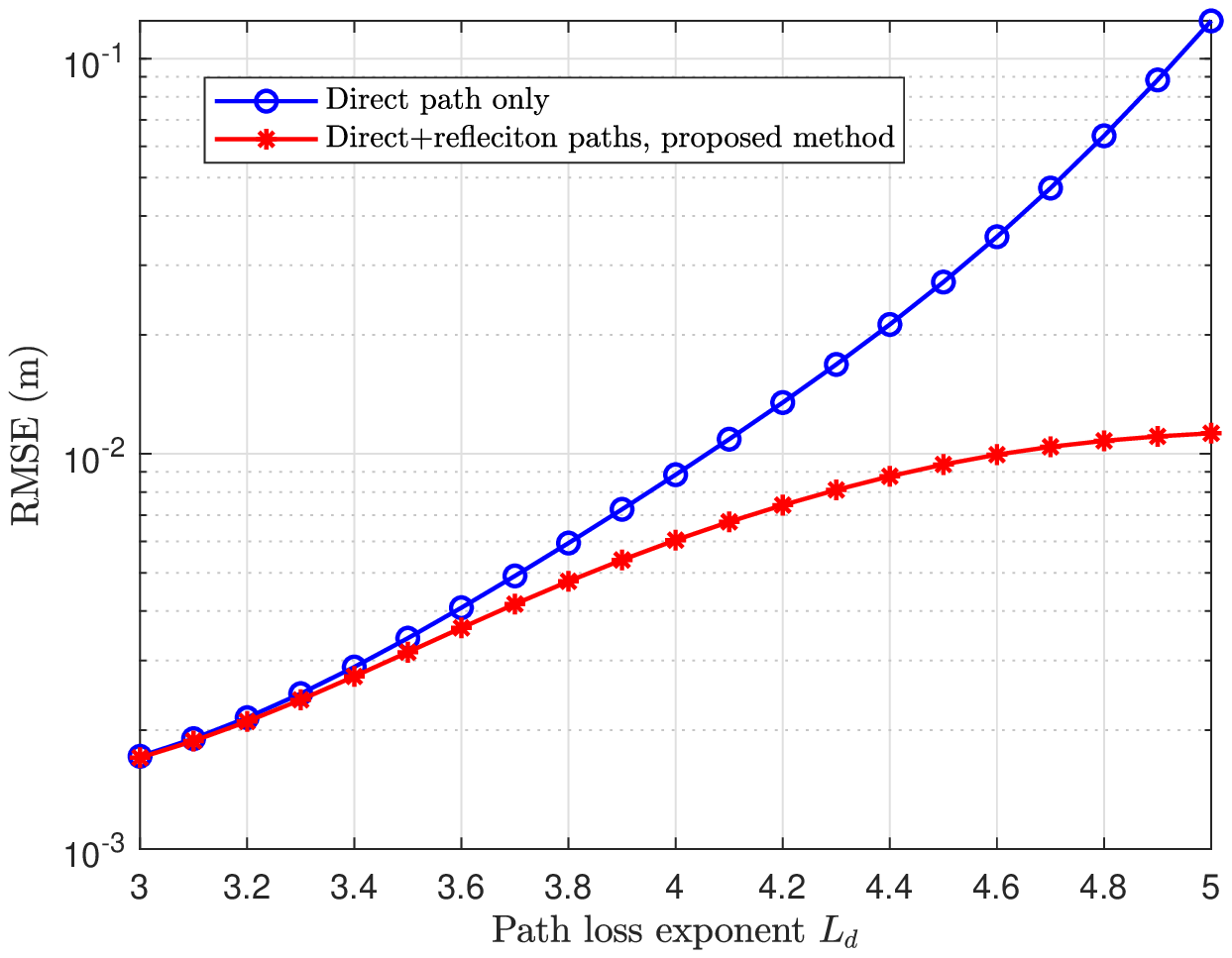}
\caption{The RMSE of UE position versus different path loss exponent of direct path.}  \label{fig:path loss exponent}
\end{figure}

In \cref{fig:path loss exponent}, we evaluate the positioning accuracy of the proposed method under different path loss exponents for the direct path. The simulation settings are the same as those in \cref{fig:fig position error} except that $L_d = \{3,3.1,\ldots,5\}$ and $1/\sigma^2=95$dB. The path loss exponent captures the blockage severity of the direct path with larger path loss exponent meaning more severe blocking. We observe from \cref{fig:path loss exponent} that when the direct path is not severely blocked ($L_d = 3$), the fusion result has similar performance as the ``direct path only'' case. As the path loss exponent increases, the positioning error using only the direct path increases, while the proposed method still produces an accurate result. This verifies the proposed positioning method can adapt to different fading scenarios of the direct path.

\subsection{Multi-UE and Multi-BS Scenarios}

\begin{figure}[!htb]
\centering
\includegraphics[width=.6\textwidth]{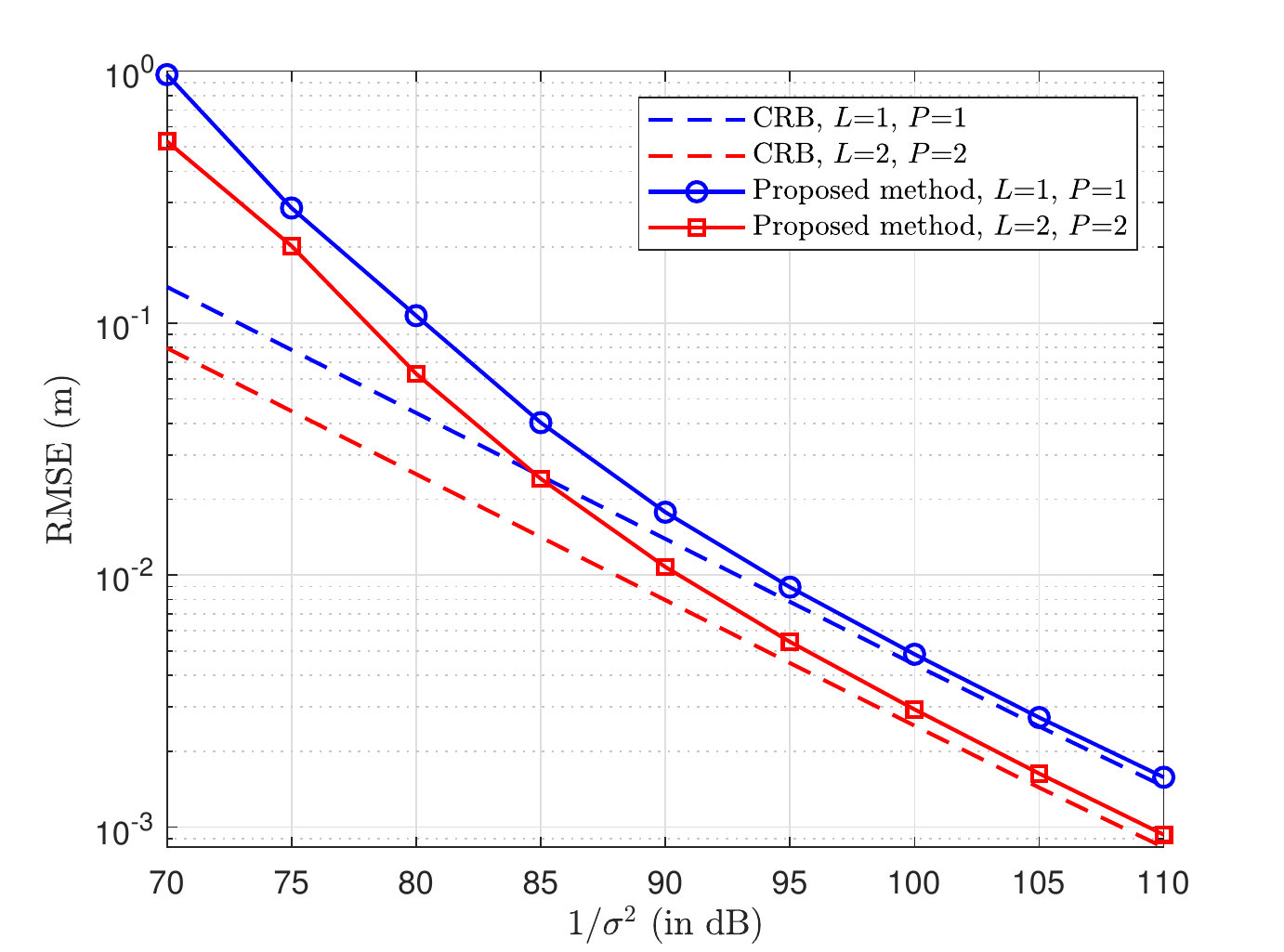}
\caption{The RMSE of UE position versus the noise level with multiple BSs and multiple UEs.} \label{fig:multi_BS}
\end{figure}

In \cref{fig:multi_BS}, we s compare the positioning accuracy achieved by using multiple BSs and multiple UEs with the scenario of single BS and single UE.
The positions of the two BSs are at $[0,0,0]$ and $[0,10,0]$. The positions of the two UEs are at $[50,10,20]$ and $[52,10,20]$. The position of the RIS is at $[30,-5,2]$.
The path loss exponents are $L_d=4.5$ and $L_r=2$.
In \cref{fig:multi_BS},  by using the techniques in \cref{sec:discussions}, the proposed positioning method in this scenario also achieves performance close to the theoretical bound. From \cref{fig:multi_BS}, when more than one BS and UE can cooperate and exchange information, the positioning accuracy can be further improved.

\section{Conclusions}\label{sec:conclusions}
In this paper, we have developed a RIS-aided positioning framework. The framework consists of first estimating the RIS-aided channel parameters from received signals, and then using these estimates to infer the UE position. Through an optimal linear combination of estimates from the direct and reflection paths, the proposed fusion method is shown via the EXIP framework to approximate the \gls{MLE} asymptotically when the estimates are independent and the number of samples is large. The advantage of our approach is computational tractability, making it amendable to real-time implementation, as compared to direct estimation of the UE position from the received signals.
Moreover, the proposed RIS-aided positioning method can be readily extended to the multi-BS and multi-user scenario.
Through simulation studies, we demonstrated
the positioning accuracy of the proposed method, which shows that it is close to the CRB and can adapt to different channel fading scenarios.

\appendices
\section{CRB Derivation}\label[Appendix]{app:CRB}
Taking the derivatives of $\widebar{\bH}_k$ \gls{wrt} the parameters of the direct path, we have
\begin{small}
\begin{align}
\frac{\partial \widebar{\bH}_k}{\partial \tau_d}&= -\iu2\pi\frac{kW}{K} \widebar{\bH}_{d,k}, \nonumber \\
\frac{\partial \widebar{\bH}_k}{\partial\opRe\{ h_d \}}&= \widebar{\bH}_{d,k}/ { h_d},~\frac{\partial \bH_k}{\partial \opIm\{h_{d}\}} = \iu \widebar{\bH}_{d,k}/ { h_d} , \nonumber \\
\frac{\partial \widebar{\bH}_k}{\partial g_{U_d}} &=  h_{d,k}\frac{\partial \ba_U(g_{U_d}, v_{U_d})}{\partial g_{U_d}}\ba_B^H(g_{B_d},v_{B_d}) ,  \nonumber \\
\frac{\partial \widebar{\bH}_k}{\partial v_{U_d}} &= h_{d,k} \frac{\partial \ba_U(g_{U_d}, v_{U_d})}{\partial v_{U_d}}\ba_B^H(g_{B_d},v_{B_d}) ,\nonumber \\
\frac{\partial \widebar{\bH}_k}{\partial g_{B_d}}&=  h_{d,k} \ba_U(g_{U_d}, v_{U_d})\frac{\partial \ba_B^H(g_{B_d},v_{B_d}) }{\partial g_{B_d}} , \nonumber \\
\frac{\partial \widebar{\bH}_k}{\partial v_{B_d}}&=h_{d,k},
\ba_U(g_{U_d}, v_{U_d})\frac{\partial \ba_B^H(g_{B_d},v_{B_d}) }{\partial v_{B_d}}, \nonumber
\end{align}
\end{small}%
and the derivative \gls{wrt} the parameters of the reflection path, we obtain
\begin{small}
\begin{align}
\frac{\partial \widebar{\bH}_{k}}{\partial \tau_{r_2,q}}&= -\iu2\pi\frac{kW}{K} \bH_{r,k,q}, \nonumber \\
\frac{\partial \widebar{\bH}_{k}}{\partial \opRe\{h_{r,q}\}}&= \bH_{r,k,q}/ { h_{r,q}},~\frac{\partial \widebar{\bH}_{k}}{\partial\opIm\{h_{r,q}\}} = \iu \bH_{r,k,q}/ { h_{r,q}} , \nonumber \\
\frac{\partial \widebar{\bH}_{k}}{\partial g_{U_r,q}} &=  h_{k} \frac{\partial \ba_U(g_{U_r,q}, v_{U,r })}{\partial g_{U_r,q}}\ba_B^H(g_{B_r,q},v_{B_r,q}),\nonumber \\
\frac{\partial \widebar{\bH}_{k}}{\partial v_{U_r,q}}&= h_{k}  \frac{\partial \ba_U(g_{U_r,q}, v_{U,r })}{\partial v_{U_r,q}}\ba_B^H(g_{B_r,q},v_{B_r,q}).\nonumber
\end{align}
\end{small}

We also have
\begin{small}
\begin{align}
\frac{\partial \ba_U(g_{U}, v_{U})}{\partial g_U} &=\left(\iu\pi \ba_{\tilde{U}}(g_{U}) \circ \bE_U\right) \otimes   \ba_{\tilde{U}}(v_{U}),\nonumber \\
\frac{\partial \ba_U(g_{U}, v_{U})}{\partial v_U}&=\ba_{\tilde{U}}(g_{U}) \otimes \left(\iu\pi \ba_{\tilde{U}}(v_{U}) \circ \bE_U\right),    \nonumber\\
\frac{\partial \ba_B(g_{B}, v_{B})}{\partial g_B} &=\left(\iu\pi \ba_{\tilde{B}}(g_{B}) \circ \bE_B\right) \otimes   \ba_{\tilde{B}}(v_{B}) ,\nonumber \\
\frac{\partial \ba_U(g_{B}, v_{B})}{\partial v_B}&=\ba_{\tilde{B}}(g_{B}) \otimes \left(\iu\pi \ba_{\tilde{B}}(v_{B}) \circ \bE_B\right),\nonumber\\
\frac{\partial \ba_R(f_{R}, v_{R})}{\partial f_R} &=\left(\iu\pi \ba_{\tilde{R}}(f_{R}) \circ \bE_R\right) \otimes   \ba_{\tilde{R}}(v_{R}), \nonumber \\
\frac{\partial \ba_R(f_{R}, v_{R})}{\partial v_R}&= \ba_{\tilde{R}}(f_{R}) \otimes \left(\iu\pi \ba_{\tilde{R}}(v_{R}) \circ \bE_R\right) ,\nonumber
\end{align}
\end{small}%
where $\bE_U=[0,1,\ldots,D^{1/4}]\T$, $\bE_B=[0,1,\ldots,N^{1/4}]\T$, and $\bE_R=[0,1,\ldots,M^{1/4}]\T$.

\section{Jacobian Matrix Derivation}\label[Appendix]{app:jaco}
In this appendix, we derive the Jacobian matrix $\bJ \in \R^{(7+5Q) \times (5+2Q)}$ used in \cref{eq:Ibxi}.
Recall that $\boldsymbol{\eta}$ in \cref{eq:bmeta} and $\bxi$ in \cref{eq:bxi}, we write $\bJ$ in the following form:
\begin{align*}
  \bJ&=\begin{bmatrix}
        \tilde{\bJ}_d\T &
        \tilde{\bJ}_{r,1}\T&
         \cdots &
       \tilde{\bJ}_{r,Q}\T
       \end{bmatrix}\T,
\end{align*}
where $\tilde{\bJ}_d=\frac{\partial \bmeta_d}{\partial\bxi\T}\in \R^{7\times (5+2Q)}$ and $\tilde{\bJ}_{r,q}=\frac{\partial \bmeta_r}{\partial\bxi\T}\in \R^{5\times (5+2Q)}$.
We first derive the Jacobian matrix of the direct path $\tilde{\bJ}_d$, whose entries are can be obtained through
\begin{small}
\begin{align}
\frac{\partial \tau_d}{\partial \bp_U\T}&=\frac{\bp_U\T}{c \| \bp_U \|_2},\ \frac{\partial\opRe\{ h_d \}}{\partial\opRe\{ h_d \}}=1,\ \frac{\partial \opIm\{h_{d}\}}{\partial \opIm\{h_{d}\}}=1\blue{,}\nonumber \\
\frac{\partial s_d}{\partial \bp_U\T}&=\frac{\partial s_d}{\partial \theta_{B_d}}\frac{\theta_{B_d}}{\partial \bp_U\T} + \frac{\partial s_d}{\phi_{B_d}}\frac{\partial \phi_{B_d}}{\partial \bp_U\T}. \nonumber
\end{align}
\end{small}%
with $s_d$ denoting any entry in $\bmeta_d$.
We have
\begin{align*}
\frac{\partial \theta_{B_d}}{\partial \bp_U\T}&=\frac{1}{ \| \bp_U \|_2^3 {\left(1-\frac{(z_U)^2}{\| \bp_U \|_2^2}\right)^{\frac{1}{2}}}}\left[x_Uz_U, y_Uz_U, -x_U^2-y_U^2\right], \\
\frac{\partial \phi_{B_d}}{\partial  \bp_U\T}&= \left[ -\frac{y_U}{x_U^2+y_U^2},\frac{x_U}{x_U^2+y_U^2} ,0 \right], \\
\frac{\partial g_{B_d}}{\partial \theta_{B_d}}&= \cos \theta_{B_d} \sin\phi_{B_d},
~\frac{\partial g_{B_d}}{\partial \phi_{B_d}}=  \sin \theta_{B_d} \cos\phi_{B_d} \\
\frac{\partial v_{B_d}}{\partial \theta_{B_d}}&= -\sin\theta_{B_d},~~\frac{\partial v_{B_d}}{\partial \phi_{B_d}}= 0,\\
\frac{\partial g_{U_d}}{\partial \theta_{B_d}}&=  [\bM_R]_{1,1}  \cos \theta _{B_d}   \cos \phi _{B_d} + [\bM_R]_{1,2}  \cos \theta _{B_d}   \sin \phi _{B_d} -   [\bM_R]_{1,3}  \sin \theta _{B_d}, \\
\frac{\partial g_{U_d}}{\partial \phi_{B_d}}&= -[\bM_R]_{1,1} \sin \phi_{B_d}\sin \theta_{B_d}
 +  [\bM_R]_{1,2} \cos \phi_{B_d}\sin \theta_{B_d},\\
\frac{\partial v_{U_d}}{\partial \theta_{B_d}}&= [\bM_R]_{2,1}  \cos \theta _{B_d} \cos \phi _{B_d} + [\bM_R]_{2,2}  \cos \theta _{B_d}   \sin \phi _{B_d}-   [\bM_R]_{2,3}  \sin \theta _{B_d},\\
\frac{\partial v_{U_d}}{\partial \phi_{B_d}}&=  -[\bM_R]_{2,1} \sin \phi_{B_d}\sin \theta_{B_d}  +  [\bM_R]_{2,2} \cos \phi_{B_d} \sin \theta_{B_d}.
\end{align*}

We next derive the Jacobian matrix of the reflection path $\tilde{\bJ}_{r,q}$, whose entries are can be obtained through
\begin{align}
\frac{\partial \tau_{r_2,q}}{\partial \bp_U\T}&=\frac{\bp_U\T-\bp_{R,q}\T}{c \| \bp_U -\bp_{R,q} \|_2},
\frac{\partial \opRe\{h_{r,q}\}}{\partial \opRe\{h_{r,q}\}}=1,\frac{\partial\opIm\{h_{r,q}\}}{\partial\opIm\{h_{r,q}\}}=1
\nonumber \\
\frac{\partial s_r}{\partial \bp_U\T}&=\frac{\partial s_r}{\partial \theta_{R_2,q}}\frac{ \theta_{R_2,q}}{\partial \bp_U\T} + \frac{\partial s_r}{ \phi_{R_2,q}}\frac{\partial  \phi_{R_2,q}}{\partial \bp_U\T},\nonumber
\end{align}
with $s_r$ denoting any entry in $\bmeta_r$.
Here, we denote $\tilde{\bp}_{U,q} = \bp_U - \bp_{R,q}  = \left[ \tilde{x}_{U,q},\tilde{y}_{U,q},\tilde{z}_{U,q} \right]\T$.
Therefore,
\begin{align*}
\frac{\partial \theta_{R_2,q}}{\partial  \bp_U\T}&=\frac{ \left[{\tilde{x}_{U,q}\tilde{z}_{U,q}}, {\tilde{y}_{U,q}\tilde{z}_{U,q}},
-{\tilde{x}_{U,q}^2-\tilde{y}_{U,q}^2}\right]}{ \|\tilde{\bp}_{U,q} \|_2^3 {\left(1-\frac{\tilde{z}_{U,q}^2}{\| \tilde{\bp}_{U,q} \|_2^2}\right)^{\frac{1}{2}}}}, \\
\frac{\partial \phi_{R_2,q}}{\partial  \bp_U\T}&= \Big[ -\frac{\tilde{y}_{U,q}}{\tilde{x}_{U,q}^2+\tilde{y}_{U,q}^2}, \frac{\tilde{x}_{U,q}}{\tilde{x}_{U,q}^2+\tilde{y}_{U,q}^2} ,0 \Big],\\
\frac{\partial g_{U_r,q}}{\partial \theta_{R_2,q}}&=[\bM_R]_{1,1}  \cos \theta _{R_2,q}\cos \phi _{R_2,q} +     [\bM_R]_{1,2}  \cos \theta _{R_2,q}\sin \phi _{R_2,q} - [\bM_R]_{1,3}  \sin \theta _{R_2,q}, \\
\frac{\partial g_{U_r,q}}{\partial \phi_{R_2,q}}&=-[\bM_R]_{1,1} \sin\theta_{R_2,q} \sin \phi_{R_2,q}+   [\bM_R]_{1,2} \cos \phi_{R_2,q} \sin \theta_{R_2,q}, \\
\frac{\partial v_{U_r,q}}{\partial \theta_{R_2,q}}&=   [\bM_R]_{2,1}  \cos \theta _{R_2,q}\cos \phi _{R_2,q} +     [\bM_R]_{2,2}  \cos \theta _{R_2,q}\sin \phi _{R_2,q} -[\bM_R]_{2,3}  \sin \theta _{R_2,q}, \\
\frac{\partial v_{U_r,q}}{\partial \phi_{R_2,q}}&= -[\bM_R]_{2,1} \sin\theta_{R_2,q} \sin \phi_{R_2,q} +   [\bM_R]_{2,2} \cos \phi_{R_2,q} \sin \theta_{R_2,q}.
\end{align*}

\section{Proof of \cref{pro:crb bound}}\label[Appendix]{app:single position}

When we utilize only the BS-UE link for UE positioning, the FIM is given by
\begin{align}
\widebar{\bC}_{\bxi_d}=\left( \bJ_d\T  \widebar{\bC}_{{\bmeta}_d}^{-1} \bJ_d \right)^{-1},
\end{align}
where $\widebar{\bC}_{{\bmeta}_d} = \left[ \bF_{\bmeta} ^{-1}\right]_{1:7,1:7}$, and $ \bJ_d  = \frac{\partial \bmeta_d }{\partial \bxi_d\T}\in \R^{7 \times 5}$. Thus, the error covariance matrix satisfies the following:
\begin{align}
\bC_{{\bp}_{U}}^{(d)} &\succeq   \left[ \widebar{\bC}_{\bxi_d}\right]_{1:3,1:3}. \label{eq: crb direct p1}
\end{align}
Since
$
\bF_{\bmeta_d}^{-1} \succeq \left[ \bF_{\bmeta} ^{-1}\right]_{1:7,1:7}
$, the following equation holds,
\begin{align}
\widebar{\bC}_{\bxi_d}\!\!=\!\!\left( \bJ_d\T \left ( \left[ \bF_{\bmeta} ^{-1}\right]_{1:7,1:7}\right)^{-1} \bJ_d \right)^{-1}
\succeq
\left( \bJ_d\T \bF_{\bmeta_d}   \bJ_d \right)^{-1}. \label{eq: crb direct p2}
\end{align}
Therefore, combining \cref{eq: crb direct p1} and \cref{eq: crb direct p2}, we have
\begin{align*}
\bC_{{\bp}_{U}}^{(d)} &\succeq  \left[ \widebar{\bC}_{\bxi_d}\right]_{1:3,1:3} \succeq  \left[ \left(\bJ_d\T \bF_{\bmeta_d}  \bJ_d \right)^{-1} \right]_{1:3,1:3}.
\end{align*}
This concludes the proof in \cref{eq:crb direct}.

Similarly, when only the $q$th RIS link is utilized for UE positioning, we can also obtain
\begin{align*}
\bC_{{\bp}_{U}}^{(r,q)} &\succeq  \left[ \widebar{\bC}_{\bxi_{r,q}}\right]_{1:3,1:3} \succeq  \left[ \left(\bJ_{r,q}\T \bF_{\bmeta_{r,q}}  \bJ_{r,q} \right)^{-1} \right]_{1:3,1:3}.
\end{align*}
This concludes the proof for \cref{eq:crb reflect}.

\section{Proof of \cref{lem:equivalence}}\label[Appendix]{app:equivalence}
It suffices to show that
\begin{small}
\begin{align*}
  \left[( \bJ\T\bF_{\hat{\bmeta}}\bJ)^{-1}
    \bJ\T
    \bF_{\hat{\bmeta}}\right]_{1:3,:}
  \begin{bmatrix}
    \bJ_d \hat{\bxi}_d \\
    \boldsymbol{0}
  \end{bmatrix}
  = \tilde{\bC}_{\bp_U}  (\tilde{\bC}_{\bp_U}^{(d)})^{-1} \hat{\bp}_{U}^{(d)}.
\end{align*}
\end{small}
\!\!\!The components for reflection paths can be proved similarly.
Thus, it is sufficient to prove the following two equalities:
\begin{align}
\left[( \bJ\T\bF_{\hat{\bmeta}}\bJ)^{-1}\right]_{1:3,:}
\left[\bJ\T
\bF_{\hat{\bmeta}}\right]_{:,1:7}\left[\bJ_d \right]_{:,1:3}
&= \tilde{\bC}_{\bp_U} (\tilde{\bC}_{\bp_U}^{(d)})^{-1} \label{eq:equivalence1}\\
\left[( \bJ\T\bF_{\hat{\bmeta}}\bJ)^{-1}\right]_{1:3,:}
\left[\bJ\T
\bF_{\hat{\bmeta}}\right]_{:,1:7}\left[\bJ_d \right]_{:,4:5}
&=\boldsymbol{0}. \label{eq:equivalence2}
\end{align}

Here, we write
\begin{align}
\bF_{\hat{\bmeta}_d}\!\! =\!\!
\begin{bmatrix}
\bF_{\hat{\bmeta}_d}^{(h)} & \bF_{\hat{\bmeta}_d}^{(h,p)} \\
\bF_{\hat{\bmeta}_d}^{(p,h)}  & \bF_{\hat{\bmeta}_d}^{(p)}
\end{bmatrix}, ~~
\bF_{\hat{\bmeta}_{r,q}} \!\!=\!\!
\begin{bmatrix}
\bF_{\hat{\bmeta}_{r,q}}^{(h)} & \bF_{\hat{\bmeta}_{r,q}}^{(h,p)} \\
\bF_{\hat{\bmeta}_{r,q}}^{(p,h)}  & \bF_{\hat{\bmeta}_{r,q}}^{(p)}
\end{bmatrix},
\end{align}
where $\bF_{\hat{\bmeta}_d}^{(h)} \in \C^{2 \times 2}$, $\bF_{\hat{\bmeta}_d}^{(p)} \in \C^{5\times 5}$, $\bF_{\hat{\bmeta}_{r,q}}^{(h)} \in \C^{2 \times 2}$, $\bF_{\hat{\bmeta}_{r,q}}^{(p)} \in \C^{3\times 3}$, and the remaining matrices have matching dimensions.
Then, for the first term in the product on \gls{LHS} of \cref{eq:equivalence1}, we can calculate
\begin{align} \label{eq:J_I_J}
 \bJ\T\bF_{\hat{\bmeta}}\bJ =
 \begin{bmatrix}
    \bZ_{1,1}&  \bZ_{1,2}\\
\bZ_{2,1}& \bZ_{2,2}
 \end{bmatrix},
\end{align}
where we denote
\begin{small}
\begin{align*}
 \bZ_{1,1}&= (\bJ_d^{(p)})\T \bF_{\hat{\bmeta}_d}^{(p)}  \bJ_d^{(p)} + \sum_{q=1}^{Q}(\bJ_{r,q}^{(p)})\T \bF_{\hat{\bmeta}_{r,q}}^{(p)}  \bJ_{r,q}^{(p)},\\
  \bZ_{1,2}&=
  \begin{bmatrix}
(\bJ_d^{(p)})\T \bF_{\hat{\bmeta}_d}^{(p,h)}  &  (\bJ_{r,1}^{(p)})\T \bI_{\hat{\bmeta}_{r,1}}^{(p,h)} & \cdots & (\bJ_{r,Q}^{(p)})\T \bI_{\hat{\bmeta}_{r,Q}}^{(p,h)}
  \end{bmatrix},\\
  \bZ_{2,1}&= \bZ_{1,2}\T,\\
  \bZ_{2,2}& = \begin{bmatrix}
   \bF_{\hat{\bmeta}_d}^{(h)} & \boldsymbol{0} & \cdots&\boldsymbol{0}\\
  \boldsymbol{0} & \bI_{\hat{\bmeta}_{r,1}}^{(h)} & \cdots&\boldsymbol{0} \\
 \vdots & \vdots & \ddots&\vdots \\
 \boldsymbol{0} &\boldsymbol{0}   & \cdots &\bI_{\hat{\bmeta}_{r,Q}}^{(h)}
               \end{bmatrix}.
\end{align*}
\end{small}
\!\!\!From \cref{eq:J_I_J}, we have
\begin{small}
\begin{align}
\!\!\!\!\left[( \bJ\T\bF_{\hat{\bmeta}}\bJ)^{-1}\right]_{1:3,:}\!\! =\!\!
\left(\bZ_{1,1}- \bZ_{1,2}\bZ_{2,2}^{-1} \bZ_{2,1}\right)^{-1}
\!\!
\begin{bmatrix}
   \bI & -\bZ_{1,2} \bZ_{2,2}^{-1}
\end{bmatrix}. \label{eq:J_I_J_inv3}
\end{align}
\end{small}
\!\!From the definitions of $\tilde{\bC}_{\bp_U}^{(d)}$ in \cref{eq:def_Cd} and $\tilde{\bC}_{\bp_U}^{(r,q)}$ in \cref{eq:def_Cr}, we can check that
$
\left[( \bJ\T\bF_{\hat{\bmeta}}\bJ)^{-1}\right]_{1:3,1:3} =\tilde{\bC}_{\bp_U}
$.
Therefore, \cref{eq:J_I_J_inv3} can be rewritten as
\begin{align}
&\left[( \bJ\T\bF_{\hat{\bmeta}}\bJ)^{-1}\right]_{1:3,:} =
\tilde{\bC}_{\bp_U}
 \begin{bmatrix}
  \bI &
 -\bZ_{1,2} \bZ_{2,2}^{-1}
  \end{bmatrix}. \label{eq:first part}
\end{align}

For the second and third terms in the product on the \gls{LHS} of \cref{eq:equivalence1}, we have
\begin{small}
\begin{align}
\left[\bJ\T
\bF_{\hat{\bmeta}}\right]_{:,1:7} \bJ_d=
 \begin{bmatrix}
   (\bJ_d^{(p)})\T \bF_{\hat{\bmeta}_d}^{(p)} \bJ_d^{(p)}
    &  (\bJ_d^{(p)})\T \bF_{\hat{\bmeta}_d}^{(p,h)}    \\
  \bF_{\hat{\bmeta}_d}^{(h,p)}\bJ_d^{(p)}  &  \bF_{\hat{\bmeta}_d}^{(h)}   \\
    \boldsymbol{0}&\boldsymbol{0}  \\
   \vdots & \vdots\\
   \boldsymbol{0} & \boldsymbol{0}
 \end{bmatrix}.  \label{eq:second part}
\end{align}
\end{small}

Combining \cref{eq:first part} and \cref{eq:second part} gives
\begin{small}
\begin{align*}
\left[( \bJ\T\bF_{\hat{\bmeta}}\bJ)^{-1}\right]_{1:3,:}
\left[\bJ\T
\bF_{\hat{\bmeta}}\right]_{:,1:7}\bJ_d
&=\tilde{\bC}_{\bp_U}
[
  (\bJ_d^{(p)})\T \bF_{\hat{\bmeta}_d}^{(p)} \bJ_d^{(p)} \!\!-\!\!(\bJ_d^{(p)})\T \bF_{\hat{\bmeta}_d}^{(p,h)}(   \bF_{\hat{\bmeta}_d}^{(h)} )^{-1}    \bF_{\hat{\bmeta}_d}^{(h,p)}\bJ_d^{(p)}~~
\boldsymbol{0}
].
\end{align*}
\end{small}
\!\!We can verify that
\begin{align*}
  (\bJ_d^{(p)})\T \bF_{\hat{\bmeta}_d}^{(p)} \bJ_d^{(p)}\!\! -\!\!(\bJ_d^{(p)})\T \bF_{\hat{\bmeta}_d}^{(p,h)}(   \bF_{\hat{\bmeta}_d}^{(h)} )^{-1}    \bF_{\hat{\bmeta}_d}^{(h,p)}\bJ_d^{(p)}
=(\tilde{\bC}_{\bp_U}^{(d)})^{-1}.
\end{align*}
Thus, we have proved  \cref{eq:equivalence1} and \cref{eq:equivalence2}. This concludes the proof.

\bibliographystyle{IEEEtran}

\bibliography{IEEEabrv,Conference_mmWave_CS}
\clearpage

\end{document}

% (\\iffalse)\n(^((?!\\fi)[\s\S])*$)\n(\\fi) 

%% file: preamble.tex
\usepackage[T1]{fontenc}
\usepackage[utf8]{inputenc}
\usepackage{amsmath,amssymb,amsfonts,mathrsfs,bm}% Typical maths resource packages
\usepackage{mathtools}
\usepackage{amsthm}
\usepackage{scalerel}
\usepackage{nicefrac}
\usepackage[shortlabels]{enumitem}
\usepackage{graphicx}
\usepackage{epstopdf}
\DeclareGraphicsExtensions{.eps,.png,.jpg,.pdf}

\usepackage{url}
\usepackage{colortbl}
\usepackage{booktabs}
\usepackage{multirow}
\usepackage[table,dvipsnames]{xcolor}
\usepackage[normalem]{ulem}
\usepackage{xparse}
\usepackage{calc}
\usepackage{etoolbox}

\makeatletter
\@ifpackageloaded{natbib}{
	\relax
}{
	\usepackage{cite}
}
\makeatother

\usepackage{array}
\newcolumntype{L}[1]{>{\raggedright\let\newline\\\arraybackslash\hspace{0pt}}m{#1}}
\newcolumntype{C}[1]{>{\centering\let\newline\\\arraybackslash\hspace{0pt}}m{#1}}
\newcolumntype{R}[1]{>{\raggedleft\let\newline\\\arraybackslash\hspace{0pt}}m{#1}}

\makeatletter
\let\MYcaption\@makecaption
\makeatother
\usepackage[font=footnotesize]{subcaption}
\makeatletter
\let\@makecaption\MYcaption
\makeatother

\usepackage{glossaries}
\makeatletter
\let\oldgls\gls
\let\oldglspl\glspl

\newcommand\fussy@ifnextchar[3]{%
	\let\reserved@d=#1%
	\def\reserved@a{#2}%
	\def\reserved@b{#3}%
	\futurelet\@let@token\fussy@ifnch}
\def\fussy@ifnch{%
	\ifx\@let@token\reserved@d
		\let\reserved@c\reserved@a
	\else
		\let\reserved@c\reserved@b
	\fi
	\reserved@c}

\renewcommand{\gls}[1]{%
\oldgls{#1}\fussy@ifnextchar.{\@checkperiod}{\@}}
\renewcommand{\glspl}[1]{%
\oldglspl{#1}\fussy@ifnextchar.{\@checkperiod}{\@}}

\newcommand{\@checkperiod}[1]{%
	\ifnum\sfcode`\.=\spacefactor\else#1\fi
}
\makeatother

\newacronym{wrt}{w.r.t.}{with respect to}
\newacronym{RHS}{R.H.S.}{right-hand side}
\newacronym{LHS}{L.H.S.}{left-hand side}
\newacronym{iid}{i.i.d.}{independent and identically distributed}
\usepackage{float}

\ifx\notloadhyperref\undefined
	\ifx\loadbibentry\undefined
		\usepackage[hidelinks,hypertexnames=false]{hyperref}
	\else
		\usepackage{bibentry}
		\makeatletter\let\saved@bibitem\@bibitem\makeatother
		\usepackage[hidelinks,hypertexnames=false]{hyperref}
		\makeatletter\let\@bibitem\saved@bibitem\makeatother
	\fi
\else
	\ifx\loadbibentry\undefined
		\relax
	\else
		\usepackage{bibentry}
	\fi
\fi

\usepackage[capitalize]{cleveref}
\crefname{equation}{}{}
\Crefname{equation}{}{}
\crefname{claim}{claim}{claims}
\crefname{step}{step}{steps}
\crefname{line}{line}{lines}
\crefname{condition}{condition}{conditions}
\crefname{dmath}{}{}
\crefname{dseries}{}{}
\crefname{dgroup}{}{}

\crefname{Problem}{Problem}{Problems}
\crefformat{Problem}{Problem~(#2#1#3)}
\crefrangeformat{Problem}{Problems~(#3#1#4) to~(#5#2#6)}

\crefname{Theorem}{Theorem}{Theorems}
\crefname{Corollary}{Corollary}{Corollaries}
\crefname{Proposition}{Proposition}{Propositions}
\crefname{Lemma}{Lemma}{Lemmas}
\crefname{Definition}{Definition}{Definitions}
\crefname{Example}{Example}{Examples}
\crefname{Assumption}{Assumption}{Assumptions}
\crefname{Remark}{Remark}{Remarks}
\crefname{Rem}{Remark}{Remarks}
\crefname{remarks}{Remarks}{Remarks}
\crefname{Appendix}{Appendix}{Appendices}
\crefname{Supplement}{Supplement}{Supplements}
\crefname{Exercise}{Exercise}{Exercises}
\crefname{Theorem_A}{Theorem}{Theorems}
\crefname{Corollary_A}{Corollary}{Corollaries}
\crefname{Proposition_A}{Proposition}{Propositions}
\crefname{Lemma_A}{Lemma}{Lemmas}
\crefname{Definition_A}{Definition}{Definitions}

\usepackage{crossreftools}
\ifx\notloadhyperref\undefined
\else
	\relax
\fi

\usepackage{algorithm,algorithmic}

\ifx\loadbreqn\undefined
	\relax
\else
	\usepackage{breqn}
\fi

\makeatletter
\def\cleartheorem#1{%
    \expandafter\let\csname#1\endcsname\relax
    \expandafter\let\csname c@#1\endcsname\relax
}
\def\clearthms#1{ \@for\tname:=#1\do{\cleartheorem\tname} }
\makeatother

\ifx\renewtheorem\undefined
	\ifx\useTheoremCounter\undefined
		\newtheorem{Theorem}{Theorem}
		\newtheorem{Corollary}{Corollary}
		\newtheorem{Proposition}{Proposition}
		\newtheorem{Lemma}{Lemma}
	\else
		\newtheorem{Theorem}{Theorem}
		
		\newtheorem{Proposition}[Theorem]{Proposition}
	\fi

	\newtheorem{Remark}{Remark}

\fi

\theoremstyle{remark}

\theoremstyle{plain}

\newcommand{\qednew}{\nobreak \ifvmode \relax \else
		\ifdim\lastskip<1.5em \hskip-\lastskip
			\hskip1.5em plus0em minus0.5em \fi \nobreak
		\vrule height0.75em width0.5em depth0.25em\fi}



\newcommand{\nn}{\nonumber\\ }

\NewDocumentCommand{\movedownsub}{e{^_}}{%
	\IfNoValueTF{#1}{%
		\IfNoValueF{#2}{^{}}% neither ^ nor _, do nothing; if no ^ but _, add ^{}
	}{%
		^{#1}% add superscript if present
	}%
	\IfNoValueF{#2}{_{#2}}% add subscript if present
}

\let\latexchi\chi
\RenewDocumentCommand{\chi}{}{\latexchi\movedownsub}

\newcommand{\iu}{\mathfrak{i}\mkern1mu}

\newcommand{\calC}{\mathcal{C}}

\newcommand{\calN}{\mathcal{N}}

\newcommand{\ba}{\mathbf{a}}
\newcommand{\bA}{\mathbf{A}}

\newcommand{\bB}{\mathbf{B}}

\newcommand{\bC}{\mathbf{C}}

\newcommand{\bD}{\mathbf{D}}

\newcommand{\bE}{\mathbf{E}}
\newcommand{\boldf}{\mathbf{f}}
\newcommand{\bF}{\mathbf{F}}

\newcommand{\bG}{\mathbf{G}}
\newcommand{\bh}{\mathbf{h}}
\newcommand{\bH}{\mathbf{H}}

\newcommand{\bI}{\mathbf{I}}

\newcommand{\bJ}{\mathbf{J}}

\newcommand{\bM}{\mathbf{M}}
\newcommand{\bn}{\mathbf{n}}
\newcommand{\bN}{\mathbf{N}}

\newcommand{\bp}{\mathbf{p}}
\newcommand{\bP}{\mathbf{P}}

\newcommand{\bQ}{\mathbf{Q}}
\newcommand{\br}{\mathbf{r}}
\newcommand{\bR}{\mathbf{R}}
\newcommand{\bs}{\mathbf{s}}

\newcommand{\bw}{\mathbf{w}}
\newcommand{\bW}{\mathbf{W}}
\newcommand{\bx}{\mathbf{x}}
\newcommand{\bX}{\mathbf{X}}
\newcommand{\by}{\mathbf{y}}

\newcommand{\bz}{\mathbf{z}}
\newcommand{\bZ}{\mathbf{Z}}

\DeclareSymbolFont{bsfletters}{OT1}{cmss}{bx}{n}
\DeclareSymbolFont{ssfletters}{OT1}{cmss}{m}{n}
\DeclareMathSymbol{\bsfGamma}{0}{bsfletters}{'000}
\DeclareMathSymbol{\ssfGamma}{0}{ssfletters}{'000}
\DeclareMathSymbol{\bsfDelta}{0}{bsfletters}{'001}
\DeclareMathSymbol{\ssfDelta}{0}{ssfletters}{'001}
\DeclareMathSymbol{\bsfTheta}{0}{bsfletters}{'002}
\DeclareMathSymbol{\ssfTheta}{0}{ssfletters}{'002}
\DeclareMathSymbol{\bsfLambda}{0}{bsfletters}{'003}
\DeclareMathSymbol{\ssfLambda}{0}{ssfletters}{'003}
\DeclareMathSymbol{\bsfXi}{0}{bsfletters}{'004}
\DeclareMathSymbol{\ssfXi}{0}{ssfletters}{'004}
\DeclareMathSymbol{\bsfPi}{0}{bsfletters}{'005}
\DeclareMathSymbol{\ssfPi}{0}{ssfletters}{'005}
\DeclareMathSymbol{\bsfSigma}{0}{bsfletters}{'006}
\DeclareMathSymbol{\ssfSigma}{0}{ssfletters}{'006}
\DeclareMathSymbol{\bsfUpsilon}{0}{bsfletters}{'007}
\DeclareMathSymbol{\ssfUpsilon}{0}{ssfletters}{'007}
\DeclareMathSymbol{\bsfPhi}{0}{bsfletters}{'010}
\DeclareMathSymbol{\ssfPhi}{0}{ssfletters}{'010}
\DeclareMathSymbol{\bsfPsi}{0}{bsfletters}{'011}
\DeclareMathSymbol{\ssfPsi}{0}{ssfletters}{'011}
\DeclareMathSymbol{\bsfOmega}{0}{bsfletters}{'012}
\DeclareMathSymbol{\ssfOmega}{0}{ssfletters}{'012}

\newcommand{\bgamma}{\bm{\gamma}}

\newcommand{\btheta}{\bm{\theta}}

\newcommand{\bmeta}{\bm{\eta}}

\newcommand{\bxi}{\bm{\xi}}

\newcommand{\bTheta}{\bm{\Theta}}

\makeatletter
\newcommand*\rel@kern[1]{\kern#1\dimexpr\macc@kerna}
\newcommand*\widebar[1]{%
  \begingroup
  \def\mathaccent##1##2{%
    \rel@kern{0.8}%
    \overline{\rel@kern{-0.8}\macc@nucleus\rel@kern{0.2}}%
    \rel@kern{-0.2}%
  }%
  \macc@depth\@ne
  \let\math@bgroup\@empty \let\math@egroup\macc@set@skewchar
  \mathsurround\z@ \frozen@everymath{\mathgroup\macc@group\relax}%
  \macc@set@skewchar\relax
  \let\mathaccentV\macc@nested@a
  \macc@nested@a\relax111{#1}%
  \endgroup
}
\makeatother

\DeclareMathOperator*{\argmax}{arg\,max}
\DeclareMathOperator*{\argmin}{arg\,min}

\DeclareMathOperator{\diag}{diag}

\DeclareMathOperator{\tr}{tr}

\DeclareMathOperator{\vect}{vec}

\newcommand{\ifbcdot}[1]{\ifblank{#1}{\cdot}{#1}}

\DeclarePairedDelimiterX\abs[1]{\lvert}{\rvert}{\ifbcdot{#1}}
\DeclarePairedDelimiterX\parens[1]{(}{)}{\ifbcdot{#1}}
\DeclarePairedDelimiterX\brk[1]{[}{]}{\ifbcdot{#1}}
\DeclarePairedDelimiterX\braces[1]{\{}{\}}{\ifbcdot{#1}}
\DeclarePairedDelimiterX\angles[1]{\langle}{\rangle}{\ifblank{#1}{\cdot,\cdot}{#1}}
\DeclarePairedDelimiterX\ip[2]{\langle}{\rangle}{\ifbcdot{#1},\ifbcdot{#2}}
\DeclarePairedDelimiterX\norm[1]{\lVert}{\rVert}{\ifbcdot{#1}}
\DeclarePairedDelimiterX\ceil[1]{\lceil}{\rceil}{\ifbcdot{#1}}
\DeclarePairedDelimiterX\floor[1]{\lfloor}{\rfloor}{\ifbcdot{#1}}

\DeclarePairedDelimiterXPP\trace[1]{\operatorname{Tr}}{(}{)}{}{\ifbcdot{#1}} % column vector
\DeclarePairedDelimiterXPP\col[1]{\operatorname{col}}{\{}{\}}{}{\ifbcdot{#1}} % column vector
\DeclarePairedDelimiterXPP\row[1]{\operatorname{row}}{\{}{\}}{}{\ifbcdot{#1}} % row vector
\DeclarePairedDelimiterXPP\erf[1]{\operatorname{erf}}{(}{)}{}{\ifbcdot{#1}}
\DeclarePairedDelimiterXPP\erfc[1]{\operatorname{erfc}}{(}{)}{}{\ifbcdot{#1}}
\DeclarePairedDelimiterXPP\KLD[2]{D}{(}{)}{}{\ifbcdot{#1}\, \delimsize\|\, \ifbcdot{#2}} % KL divergence
\DeclarePairedDelimiterXPP\op[2]{\operatorname{#1}}{(}{)}{}{#2} % general operator

\newcommand{\T}{^{\intercal}}% transpose notation
\DeclarePairedDelimiterXPP\indicate[1]{{\bf 1}}{\{}{\}}{}{\ifbcdot{#1}}

\providecommand\given{}
\DeclarePairedDelimiterX\Set[2]\{\}{%
\renewcommand\given{\SetSymbol[\delimsize]{#1}}
#2
}
\DeclarePairedDelimiterX\Setc[1]\{\}{%
\renewcommand\given{\SetSymbol{:}}
#1
}

\NewDocumentCommand\set{s o m}{%
	\IfBooleanTF#1%
	{\IfValueTF{#2}{\Set*{#2}{#3}}{\Setc*{#3}}}%
	{\IfValueTF{#2}{\Set{#2}{#3}}{\Setc{#3}}}%
}

\NewDocumentCommand{\evalat}{ s O{\big} m e{_^} }{%
\IfBooleanTF{#1}%
{\left. #3 \right|}{#3#2|}%
\IfValueT{#4}{_{#4}}%
\IfValueT{#5}{^{#5}}%
}

\NewDocumentCommand \ifcondp {m m} {%
	#1%
	\IfValueT{#2}{\given #2}%
}

\providecommand\given{}
\DeclarePairedDelimiterXPP\cprob[1]{}(){}{
\renewcommand\given{\nonscript\,\delimsize\vert\allowbreak\nonscript\,\mathopen{}}
\ifcondp#1
}
\DeclarePairedDelimiterXPP\cexp[1]{}[]{}{
\renewcommand\given{\nonscript\,\delimsize\vert\allowbreak\nonscript\,\mathopen{}}
\ifcondp#1
}

\DeclareDocumentCommand \P { s e{_^} >{\SplitArgument{ 1 }{ @| }}d() g } {%
	\mathbb{P}%
	\IfBooleanTF{#1}%
		{
			\IfValueT{#2}{_{#2}}%
			\IfValueT{#3}{^{#3}}%
			\IfValueTF{#5}%
				{\cprob{#4 \given #5}}%
				{\IfValueT{#4}{\cprob{#4}}}%
		}%
		{
			\IfValueT{#2}{_{#2}}%
			\IfValueT{#3}{^{#3}}%
			\IfValueTF{#5}%
				{\cprob*{#4 \given #5}}%
				{\IfValueT{#4}{\cprob*{#4}}}%
		}%
}

\DeclareDocumentCommand \E { s e{_^} >{\SplitArgument{ 1 }{ @| }}d[] g } {%
	\mathbb{E}%
	\IfBooleanTF{#1}%
		{
			\IfValueT{#2}{_{#2}}%
			\IfValueT{#3}{^{#3}}%
			\IfValueTF{#5}%
				{\cexp{#4 \given #5}}%
				{\IfValueT{#4}{\cexp{#4}}}%
		}%
		{
			\IfValueT{#2}{_{#2}}%
			\IfValueT{#3}{^{#3}}%			
			\IfValueTF{#5}%
				{\cexp*{#4 \given #5}}%
				{\IfValueT{#4}{\cexp*{#4}}}%
		}%
}

\ExplSyntaxOn
\NewDocumentCommand \dist {m o o} {%
\mathrm{#1}\left(%
	\IfValueT{#3}{%
		\tl_if_blank:nTF{ #3 }{\cdot\, \middle|\, }{#3\, \middle|\, }%
	}
	\IfValueT{#2}{#2}%
\right)%
}
\ExplSyntaxOff

\NewDocumentCommand {\cbrace} {t+ D[]{black} D(){\widthof{#5}} m m } {%
	\begingroup%
		\color{#2}
		\IfBooleanTF{#1}{%
			\overbrace{#4}^%
		}{
			\underbrace{#4}_%
		}%
		{\parbox[c]{#3}{\centering\footnotesize{#5}}}%
	\endgroup% 
}

\let\oldforall\forall
\renewcommand{\forall}{\oldforall \, }

\let\oldexist\exists
\renewcommand{\exists}{\oldexist \, }

\graphicspath{{./Figures/}{./figures/}}
\DeclareDocumentCommand{\includeCroppedPdf}{ o O{./Figures/} m }{
	\IfFileExists{#2#3-crop.pdf}{}{%
		\immediate\write18{pdfcrop #2#3.pdf #2#3-crop.pdf}}%
	\includegraphics[#1]{#2#3-crop.pdf}
}

\makeatletter
\newcommand*{\addFileDependency}[1]{% argument=file name and extension
  \typeout{(#1)}
  \@addtofilelist{#1}
  \IfFileExists{#1}{}{\typeout{No file #1.}}
}
\makeatother

\definecolor{gray90}{gray}{0.9}

\ifx\nohighlights\undefined
	
	\newcommand{\blue}[1]{{{\color{blue} #1}}}
	\newcommand{\msout}[1]{\text{\color{green} \sout{\ensuremath{#1}}}}
	\newcommand{\del}[1]{{\color{green}\ifmmode \msout{#1}\else\sout{#1}\fi}}
\else
	
	\newcommand{\blue}[1]{#1}
	\newcommand{\msout}[1]{#1}
	\newcommand{\del}[1]{#1}
\fi

\newcommand{\hhide}[1]{}

\newcommand{\txp}[2]{\texorpdfstring{#1}{#2}}

\ifx\diagnoselabel\undefined
	\relax
\else
	\makeatletter
	\def\@testdef #1#2#3{%
		\def\reserved@a{#3}\expandafter \ifx \csname #1@#2\endcsname
			\reserved@a  \else
			\typeout{^^Jlabel #2 changed:^^J%
				\meaning\reserved@a^^J%
				\expandafter\meaning\csname #1@#2\endcsname^^J}%
			\@tempswatrue \fi}
	\makeatother
\fi

%% file: mydefinitions.tex
  \def\R{{\mathbb{R}}} \def\C{{\mathbb{C}}}   %\def\E{{\mathbb{E}}}
\def\Eb{{\mathbb{E}}}
\newcommand{\beq}{\begin{eqnarray}}
\newcommand{\eeq}{\end{eqnarray}}

  \def\cC{{\mathcal{C}}}

 \def\cN{{\mathcal{N}}}

           \def\lA{\left\|}     \def\rA{\right\|}

\newcommand{\opRe}{\operatorname{Re}}
\newcommand{\opIm}{\operatorname{Im}}
\newcommand{\Col}{\operatorname{Col}}

\newacronym{MLE}{MLE}{maximum likelihood estimate}
\newacronym{BLUE}{BLUE}{best linear unbiased estimator}

%% file: Position_latex-one-column.bbl
\begin{thebibliography}{10}
\providecommand{\url}[1]{#1}
\csname url@samestyle\endcsname
\providecommand{\newblock}{\relax}
\providecommand{\bibinfo}[2]{#2}
\providecommand{\BIBentrySTDinterwordspacing}{\spaceskip=0pt\relax}
\providecommand{\BIBentryALTinterwordstretchfactor}{4}
\providecommand{\BIBentryALTinterwordspacing}{\spaceskip=\fontdimen2\font plus
\BIBentryALTinterwordstretchfactor\fontdimen3\font minus
  \fontdimen4\font\relax}
\providecommand{\BIBforeignlanguage}[2]{{%
\expandafter\ifx\csname l@#1\endcsname\relax
\typeout{** WARNING: IEEEtran.bst: No hyphenation pattern has been}%
\typeout{** loaded for the language `#1'. Using the pattern for}%
\typeout{** the default language instead.}%
\else
\language=\csname l@#1\endcsname
\fi
#2}}
\providecommand{\BIBdecl}{\relax}
\BIBdecl

\bibitem{DongSur2020}
S.~{Gong}, X.~{Lu}, D.~T. {Hoang}, D.~{Niyato}, L.~{Shu}, D.~I. {Kim}, and
  Y.~C. {Liang}, ``Toward smart wireless communications via intelligent
  reflecting surfaces: A contemporary survey,'' \emph{IEEE Commun. Surv.
  Tutor.}, vol.~22, no.~4, pp. 2283--2314, 2020.

\bibitem{Basar2019WirelessAccess}
E.~Basar, M.~Di~Renzo, J.~De~Rosny, M.~Debbah, M.-S. Alouini, and R.~Zhang,
  ``Wireless communications through reconfigurable intelligent surfaces,''
  \emph{IEEE Access}, vol.~7, pp. 116\,753--116\,773, 2019.

\bibitem{RenzoSur2020}
M.~{Di Renzo}, K.~{Ntontin}, J.~{Song}, F.~H. {Danufane}, X.~{Qian},
  F.~{Lazarakis}, J.~{De Rosny}, D.~T. {Phan-Huy}, O.~{Simeone}, R.~{Zhang},
  M.~{Debbah}, G.~{Lerosey}, M.~{Fink}, S.~{Tretyakov}, and S.~{Shamai},
  ``Reconfigurable intelligent surfaces vs. relaying: Differences,
  similarities, and performance comparison,'' \emph{IEEE Open J. Commun. Soc.},
  vol.~1, pp. 798--807, 2020.

\bibitem{EmilBeat2020}
E.~{Björnson}, .~{Özdogan}, and E.~G. {Larsson}, ``Intelligent reflecting
  surface versus decode-and-forward: How large surfaces are needed to beat
  relaying?'' \emph{IEEE Wireless Commun. Lett.}, vol.~9, no.~2, pp. 244--248,
  2020.

\bibitem{BoulSur2020}
A.~A. {Boulogeorgos} and A.~{Alexiou}, ``Performance analysis of reconfigurable
  intelligent surface-assisted wireless systems and comparison with relaying,''
  \emph{IEEE Access}, vol.~8, pp. 94\,463--94\,483, 2020.

\bibitem{Liang2019Large}
Y.-C. Liang, R.~Long, Q.~Zhang, J.~Chen, H.~V. Cheng, and H.~Guo, ``Large
  intelligent surface/antennas ({LISA}): Making reflective radios smart,''
  \emph{Journal of Communications and Information Networks}, vol.~4, no.~2, pp.
  40--50, 2019.

\bibitem{ChenSum2019}
W.~{Chen}, X.~{Ma}, Z.~{Li}, and N.~{Kuang}, ``Sum-rate maximization for
  intelligent reflecting surface based terahertz communication systems,'' in
  \emph{2019 IEEE/CIC International Conference on Communications Workshops in
  China (ICCC Workshops)}, 2019, pp. 153--157.

\bibitem{ZhaoRIS2021}
M.~M. {Zhao}, Q.~{Wu}, M.~J. {Zhao}, and R.~{Zhang}, ``Intelligent reflecting
  surface enhanced wireless networks: Two-timescale beamforming optimization,''
  \emph{IEEE Trans. Wireless Commun.}, vol.~20, no.~1, pp. 2--17, 2021.

\bibitem{NadeemCE2020}
Q.~U.~A. {Nadeem}, H.~{Alwazani}, A.~{Kammoun}, A.~{Chaaban}, M.~{Debbah}, and
  M.~S. {Alouini}, ``Intelligent reflecting surface-assisted multi-user miso
  communication: Channel estimation and beamforming design,'' \emph{IEEE Open
  J. Commun. Soc.}, vol.~1, pp. 661--680, 2020.

\bibitem{Basar2019Tran}
E.~Basar, ``Transmission through large intelligent surfaces: A new frontier in
  wireless communications,'' in \emph{2019 European Conference on Networks and
  Communications (EuCNC)}, 2019, pp. 112--117.

\bibitem{Atapattu2020RIS}
S.~Atapattu, R.~Fan, P.~Dharmawansa, G.~Wang, J.~Evans, and T.~A. Tsiftsis,
  ``Reconfigurable intelligent surface assisted two–way communications:
  Performance analysis and optimization,'' \emph{{IEEE} Trans. Commun.},
  vol.~68, no.~10, pp. 6552--6567, 2020.

\bibitem{Ferreira2020Bit}
R.~C. Ferreira, M.~S.~P. Facina, F.~A.~P. De~Figueiredo, G.~Fraidenraich, and
  E.~R. De~Lima, ``Bit error probability for large intelligent surfaces under
  double-{Nakagami} fading channels,'' \emph{IEEE Open J. Commun. Soc.},
  vol.~1, pp. 750--759, 2020.

\bibitem{Yassin2017Localization}
A.~Yassin, Y.~Nasser, M.~Awad, A.~Al-Dubai, R.~Liu, C.~Yuen, R.~Raulefs, and
  E.~Aboutanios, ``Recent advances in indoor localization: A survey on
  theoretical approaches and applications,'' \emph{IEEE Commun. Surv. Tutor.},
  vol.~19, no.~2, pp. 1327--1346, 2017.

\bibitem{WEN2019Survey}
F.~Wen, H.~Wymeersch, B.~Peng, W.~P. Tay, H.~C. So, and D.~Yang, ``A survey on
  {5G} massive {MIMO} localization,'' \emph{Digit. Signal Process.}, vol.~94,
  pp. 21--28, 2019.

\bibitem{Wymeersch2017mmWave}
H.~Wymeersch, G.~Seco-Granados, G.~Destino, D.~Dardari, and F.~Tufvesson,
  ``{5G} mmwave positioning for vehicular networks,'' \emph{IEEE Wirel.
  Commun.}, vol.~24, no.~6, pp. 80--86, 2017.

\bibitem{3gpp.21.916}
\BIBentryALTinterwordspacing
3GPP, ``{Release description; Release 16},'' {3rd Generation Partnership
  Project (3GPP)}, Technical Specification (TS) 21.916, 2021. [Online].
  Available:
  \url{https://portal.3gpp.org/desktopmodules/Specifications/SpecificationDetails.aspx?specificationId=3493}
\BIBentrySTDinterwordspacing

\bibitem{Vari2014mmV}
M.~Vari and D.~Cassioli, ``mmwaves rssi indoor network localization,'' in
  \emph{2014 IEEE International Conference on Communications Workshops (ICC)},
  2014, pp. 127--132.

\bibitem{Lin2018indoor}
Z.~Lin, T.~Lv, and P.~T. Mathiopoulos, ``{3-D} indoor positioning for
  millimeter-wave massive {MIMO} systems,'' \emph{{IEEE} Trans. Commun.},
  vol.~66, no.~6, pp. 2472--2486, 2018.

\bibitem{Parasas2018ML}
K.~N. R. S.~V. Prasad, E.~Hossain, and V.~K. Bhargava, ``Machine learning
  methods for {RSS}-based user positioning in distributed massive {MIMO},''
  \emph{IEEE Trans. Wireless Commun.}, vol.~17, no.~12, pp. 8402--8417, 2018.

\bibitem{Fas2019mmW}
A.~Fascista, A.~Coluccia, H.~Wymeersch, and G.~Seco-Granados, ``Millimeter-wave
  downlink positioning with a single-antenna receiver,'' \emph{IEEE Trans.
  Wireless Commun.}, vol.~18, no.~9, pp. 4479--4490, 2019.

\bibitem{Zhou2017low}
Z.~Zhou, J.~Fang, L.~Yang, H.~Li, Z.~Chen, and R.~S. Blum, ``Low-rank tensor
  decomposition-aided channel estimation for millimeter wave {MIMO}-{OFDM}
  systems,'' \emph{{IEEE} J. Sel. Areas Commun.}, vol.~35, no.~7, pp.
  1524--1538, 2017.

\bibitem{Wen2021Pos}
F.~Wen, J.~Kulmer, K.~Witrisal, and H.~Wymeersch, ``{5G} positioning and
  mapping with diffuse multipath,'' \emph{IEEE Trans. Wireless Commun.},
  vol.~20, no.~2, pp. 1164--1174, 2021.

\bibitem{Alkhateeb2014Channel}
A.~Alkhateeb, O.~El~Ayach, G.~Leus, and R.~W. Heath, ``Channel estimation and
  hybrid precoding for millimeter wave cellular systems,'' \emph{IEEE J. Sel.
  Topics Signal Process.}, vol.~8, no.~5, pp. 831--846, 2014.

\bibitem{Zhang2018RIP}
W.~Zhang, T.~Kim, D.~J. Love, and E.~Perrins, ``Leveraging the restricted
  isometry property: Improved low-rank subspace decomposition for hybrid
  millimeter-wave systems,'' \emph{{IEEE} Trans. Commun.}, vol.~66, no.~11, pp.
  5814--5827, 2018.

\bibitem{Shah2018Position}
A.~Shahmansoori, G.~E. Garcia, G.~Destino, G.~Seco-Granados, and H.~Wymeersch,
  ``Position and orientation estimation through millimeter-wave {MIMO} in {5G}
  systems,'' \emph{IEEE Trans. Wireless Commun.}, vol.~17, no.~3, pp.
  1822--1835, 2018.

\bibitem{Hu2018Beyond}
S.~Hu, F.~Rusek, and O.~Edfors, ``Beyond massive {MIMO}: The potential of
  positioning with large intelligent surfaces,'' \emph{IEEE Trans. Signal
  Process.}, vol.~66, no.~7, pp. 1761--1774, 2018.

\bibitem{Aleg2019CRLB}
J.~V. Alegría and F.~Rusek, ``Cramér-rao lower bounds for positioning with
  large intelligent surfaces using quantized amplitude and phase,'' in
  \emph{2019 53rd Asilomar Conference on Signals, Systems, and Computers},
  2019, pp. 10--14.

\bibitem{Ji2020large}
J.~He, H.~Wymeersch, L.~Kong, O.~Silvén, and M.~Juntti, ``Large intelligent
  surface for positioning in millimeter wave {MIMO} systems,'' in \emph{2020
  IEEE 91st Vehicular Technology Conference (VTC2020-Spring)}, 2020, pp. 1--5.

\bibitem{Ma2021indoor}
T.~Ma, Y.~Xiao, X.~Lei, W.~Xiong, and Y.~Ding, ``Indoor localization with
  reconfigurable intelligent surface,'' \emph{IEEE Commun. Lett.}, vol.~25,
  no.~1, pp. 161--165, 2021.

\bibitem{Elzanaty2021Re}
A.~Elzanaty, A.~Guerra, F.~Guidi, and M.~Alouini, ``Reconfigurable intelligent
  surfaces for localization: Position and orientation error bounds,''
  \emph{IEEE Trans. Signal Process.}, pp. 1--1, 2021.

\bibitem{zhang2021toward}
H.~Zhang, H.~Zhang, B.~Di, K.~Bian, Z.~Han, and L.~Song, ``Towards ubiquitous
  positioning by leveraging reconfigurable intelligent surface,'' \emph{IEEE
  Commun. Lett.}, vol.~25, no.~1, pp. 284--288, 2021.

\bibitem{Zhang2021Meta}
------, ``Metalocalization: Reconfigurable intelligent surface aided multi-user
  wireless indoor localization,'' \emph{IEEE Trans. Wireless Commun.}, pp.
  1--1, 2021.

\bibitem{Lin2021Channel}
Y.~Lin, S.~Jin, M.~Matthaiou, and X.~You, ``Channel estimation and user
  localization for {IRS}-assisted {MIMO-OFDM} systems,'' \emph{IEEE Trans.
  Wireless Commun.}, pp. 1--1, 2021.

\bibitem{Jensen2020}
T.~L. {Jensen} and E.~{De Carvalho}, ``An optimal channel estimation scheme for
  intelligent reflecting surfaces based on a minimum variance unbiased
  estimator,'' in \emph{2020 IEEE International Conference on Acoustics, Speech
  and Signal Processing (ICASSP)}, 2020, pp. 5000--5004.

\bibitem{De2021Channel}
G.~T. de~Araújo, A.~L.~F. de~Almeida, and R.~Boyer, ``Channel estimation for
  intelligent reflecting surface assisted {MIMO} systems: A tensor modeling
  approach,'' \emph{IEEE J. Sel. Topics Signal Process.}, vol.~15, no.~3, pp.
  789--802, 2021.

\bibitem{Zhang2021Cost}
W.~Zhang and W.~P. Tay, ``Cost-efficient {RIS}-aided channel estimation via
  rank-one matrix factorization,'' \emph{IEEE Wireless Commun. Lett.}, pp.
  1--1, 2021.

\bibitem{He2020Cascaded}
Z.-Q. He and X.~Yuan, ``Cascaded channel estimation for large intelligent
  metasurface assisted massive {MIMO},'' \emph{IEEE Wireless Commun. Lett.},
  vol.~9, no.~2, pp. 210--214, 2020.

\bibitem{chen2019channel}
J.~Chen, Y.-C. Liang, H.~V. Cheng, and W.~Yu, ``Channel estimation for
  reconfigurable intelligent surface aided multi-user {MIMO} systems,''
  \emph{arXiv preprint arXiv:1912.03619}, 2019.

\bibitem{stoica1989reparametrization}
P.~Stoica and T.~S{\"o}derstr{\"o}m, ``On reparametrization of loss functions
  used in estimation and the invariance principle,'' \emph{Signal Process.},
  vol.~17, no.~4, pp. 383--387, 1989.

\bibitem{Jing2020Positioning}
J.~Zhang, Z.~Zheng, Z.~Fei, and X.~Bao, ``Positioning with dual reconfigurable
  intelligent surfaces in millimeter-wave {MIMO} systems,'' in \emph{2020
  IEEE/CIC International Conference on Communications in China (ICCC)}, 2020,
  pp. 800--805.

\bibitem{wymeersch2009cooperative}
H.~Wymeersch, J.~Lien, and M.~Z. Win, ``Cooperative localization in wireless
  networks,'' \emph{Proc. IEEE}, vol.~97, no.~2, pp. 427--450, 2009.

\bibitem{xu2015distributed}
W.~Xu, F.~Quitin, M.~Leng, W.~P. Tay, and S.~G. Razul, ``Distributed
  localization of a {RF} target in nlos environments,'' \emph{{IEEE} J. Sel.
  Areas Commun.}, vol.~33, no.~7, pp. 1317--1330, 2015.

\bibitem{conti2012network}
A.~Conti, M.~Guerra, D.~Dardari, N.~Decarli, and M.~Z. Win, ``Network
  experimentation for cooperative localization,'' \emph{{IEEE} J. Sel. Areas
  Commun.}, vol.~30, no.~2, pp. 467--475, 2012.

\end{thebibliography}
